\documentclass[12pt]{article}

\usepackage[leqno]{amsmath}
\usepackage{amsthm}
\usepackage{graphicx}
\usepackage{setspace}
\usepackage[round]{natbib}
\usepackage{amssymb}
\usepackage{hyperref}
\usepackage{multirow}
\usepackage{rotating}
\usepackage{amsfonts}
\usepackage{mathtools}
\usepackage{pdflscape}
\usepackage{longtable}
\usepackage{chngcntr}
\usepackage{tikz}
\usepackage[capposition=top]{floatrow}
\usepackage{geometry}
\usepackage{booktabs}
\usepackage{tabularx}

\usepackage{bm}
\usepackage{comment}
\usepackage{soul}

\newtheorem{theorem}{Theorem}[section]
\newtheorem{assumption}{Assumption}[section]
\newtheorem{corollary}[theorem]{Corollary}
\newtheorem{lemma}[theorem]{Lemma}

\theoremstyle{definition}
\newtheorem{definition}{Definition}[section]
\newtheorem{example}{Example}
\newtheorem*{example*1}{Example 1}
\theoremstyle{remark}

\numberwithin{equation}{section}
\newcommand{\abs}[1]{\left\vert#1\right\vert}

\newcommand{\Real}{\mathbb R}

\newcommand{\parZ}{(\bm{Z})}

\newcommand{\commentOut}[1]{}

\def\Perp{\perp\!\!\!\perp}

\newcommand{\dpart}[2]{\frac{\partial #1}{\partial #2}}
\newcommand{\ddpart}[3]{\frac{\partial^2 #1}{\partial #2 \partial #3}}

\def\uniset{{\rm 1\kern-.40em 1}}

\begin{document}


\onehalfspacing

\title{\large\textbf{Identifying Effects of Multivalued Treatments}\footnote{We are grateful to the co-editor and to
four anonymous referees for their comments. We also thank  St\'ephane Bonhomme, Eric Gautier, Joe Hotz, Thierry Magnac, Lars Nesheim, Rodrigo Pinto, Adam Rosen, Christoph Rothe, Azeem Shaikh, Alex Torgovitsky, Ed Vytlacil, and especially Jim Heckman for their very useful suggestions. We also benefited from  comments of seminar audiences  in Cambridge, Chicago, Duke, Georgetown, Harvard, MIT, NYU, UCL, and Yale.
We would like to thank Junlong Feng and Cameron LaPoint for proofreading the paper.
This research has received financial support from the European Research Council under the European Community's Seventh Framework Program FP7/2007-2013 grant agreement No. 295298-DYSMOIA and 
under the Horizon 2020 Framework Program grant agreement No. 646917-ROMIA.}
}
\author{\textbf{Sokbae Lee}\footnote{Columbia University and Institute for Fiscal Studies, sl3841@columbia.edu.} \and \textbf{Bernard Salani\'e}\footnote{Columbia University, bsalanie@columbia.edu.}}
\date{April 20, 2018}

\maketitle

\begin{abstract}
	Multivalued treatment models have typically been studied under restrictive assumptions: ordered choice, and more recently unordered monoto\-ni\-ci\-ty.  We show how   treatment effects can be identified in  a more general class of models  that allows for multidimensional unobserved heterogeneity. Our results rely on two main assumptions: treatment assignment must be a measurable function of threshold-crossing rules, and enough continuous instruments must be available. We illustrate our approach for  several classes of models. \\

\end{abstract}

\clearpage

\section{Introduction}
 Since the seminal work of \cite{heckman1979sample}, selection problems have been  one of the main themes in both empirical economics and econometrics.
One popular  approach in the literature  is to rely on  instruments 
to uncover the patterns of the 
self-selection into different levels of treatments, and thereby to identify treatment effects. 
The main branches of this literature are  the local average treatment effect (LATE) framework of \cite{LATE1994} and the local instrumental variables (LIV) framework of \cite{MIV2005}. 

The LATE and LIV frameworks emphasize different parameters of interest and suggest different estimation methods. However, they both focus on binary treatments, and   restrict selection mechanisms to be ``monotonic''.  \cite{vytlacil2002independence} establishes that the LATE and LIV approaches rely on the same  monotonicity assumption. For binary treatment models, these approaches   require that selection into treatment be governed by a single index crossing a threshold.

Many real-world selection problems  are not adequately described by
 single-crossing models. The literature has developed ways of dealing with less restrictive models of assignment to treatment.  \cite{angrist1995two} analyze ordered choice models. \cite{HUV2006,HUV2008} show how (depending on restrictions and instruments) a variety of   treatment effects can be identified in  discrete choice models that are additively separable in instruments and errors. More recently, \cite{heckmanpinto-pdt} define an ``unordered monotonicity'' condition that is weaker than monotonicity when treatment is multivalued. They show that given unordered monotonicity, several treatment effects can be identified.  
  
Even the most generally applicable of these approaches can still  only deal with models of treatment that are formally analogous to  an additively separable  discrete choice model, as proved in Section~6 of \cite{heckmanpinto-pdt}.   The key condition is that the data contain changes in instruments that create only {\em  one-way flows\/} in or out of the treatment cells the analyst is interested in.  In binary treatment models, this is exactly the meaning of monotonicity: there cannot be both compliers and defiers, so that LATE estimates the average treatment effect on compliers\footnote{\cite{deChaisemartin} shows that  under a weaker condition, LATE estimates the average treatment effect on a specific subset of the compliers.}.  Things are somewhat more complex in  multivalued treatment models. Unless selection only depends on one function of the instruments, there exist changes in instruments that generate two-way flows in and out of any treatment cell. Unordered monotonicity requires that we observe \emph{some\/} changes in instruments  that only induce one way-flows.

 This is still too restrictive for important applications. For instance, many transfer programs  (or many educational tests) rely on several criteria and combine them in complex ways to assign agents to treatments; and agents add their own objectives and criteria to the list.   An additively separable  discrete choice model  may not describe such  a selection mechanism. 
To see this, start from a very simple and useful application: the  double hurdle model, which treats agents only if  each of {\em two\/} indices passes a threshold\footnote{See e.g. \cite{Poirier80} for a parametric version of this model.}. While this is a binary treatment model, the existence of two thresholds makes it non-monotonic: if a change in instruments increases a threshold but reduces the other, some agents  move into the treatment group and some  move out of it. 

  The double hurdle model is still  unordered monotonic, as any change in instruments that moves the two thresholds in the same direction only creates one-way flows. Now let us change the structure of the model slightly: there are still two thresholds, but we only treat agents who are above one threshold and below the other. 
As we will  see in Section~\ref{sec:model},  {\em any\/} change in instruments that moves both thresholds  
 generates two-way flows, and standard approaches to identification fail.
This  model of  {\em selection with two-way flows\/}  cannot be represented by a  discrete choice model; it is formally equivalent to a discrete choice model with three alternatives in which the analyst only observes  partitioned choices (e.g. the analyst only observes whether  alternative~2 is chosen or not).  Our identification results apply to this variant of the double hurdle model, and to all treatment models generated by a finite family of threshold-crossing rules. In fact, one way to describe our contribution is that it encompasses all   additively  separable discrete choice  models in which the analyst only observes a partition of the set of alternatives.

To illustrate the applicability of our framework, assume that assignment to treatment can be described by a random utility model of choice. Now imagine that, as is common in practice,  the analyst only  observes choices between {\em sets\/} of treatments: e.g., various vocational programs have been aggregated into a  ``training'' category in her dataset.  Our methods  allow identification of the effect of these different training programs on outcomes, provided that   continuous instruments shift their  mean utilities. Variables such as distance to the locations of the training centers  or other components of the ``full cost'' of treatment  could serve as instruments in this application.
For another example, consider a dynamic sequence of treatments such as the curriculum of a college student or the career of a worker. This could be represented as a ``decision tree'' in which various threshold-crossing rules govern the path of the individual through time. Again, this type of model can be analyzed using the techniques in this paper.  Here we could use measures of performances of the worker, or the grades of the student, as (quasi) continuous instruments  in order to infer the effect of each of the possible paths on outcomes.
We study related examples  more formally  in Section~\ref{sec:apps}.

 Our analysis allows   selection to be determined by a vector of threshold-crossing rules. Each of these rules compares  a scalar unobservable to a threshold; these unobservables can be correlated with each other and with potential outcomes.  We proceed in two steps. First assume that the thresholds are known to the analyst. We use
their values as
 control variables to deal with
multidimensional unobserved heterogeneity.  One important difference with the unidimensional case is that in our setting LATE-type estimators can only recover a mixture of causal parameters on groups that cross different thresholds, and are therefore harder to interpret. 
We establish conditions under which one can identify   a generalized version of  the marginal treatment effects (MTE) of \cite{MIV2005}, as well as the probability distribution
of unobservables governing the selection mechanism, and more aggregated treatment effects such as the average treatment effect (ATE),  quantile treatment effects, the average treatment effect on the treated (ATT), and the policy-relevant treatment effect (PRTE). 

Since thresholds often are not known a priori, the second step requires identifying them from the data. 
This is highly model-specific and the family of models encompassed in this paper is too large and diverse to allow for a general result. We  limit our discussion to a few applications; in particular, we provide what we believe are new identification theorems for the double-hurdle model.

We  give a detailed comparison of our paper to the existing literature in Section~\ref{sec:literature}.  Let us here mention  a few points in which our paper differs from the literature. Unlike  \cite{imbens2000role}, \cite{hirano2004propensity}, \cite{cattaneo2010efficient}, and \cite{Yang-et-al:16},  we   allow for selection on unobservables.  \cite{gautier-hoderlein} study  binary treatment when selection is driven by a rule that is linear in a vector of  unobservable heterogeneity.
\cite{Lewbel2016}  consider a different non-monotonic rule for binary treatment to identify the average treatment effect.
These two papers  break monotonicity in different ways than ours. We focus on the point identification of marginal treatment effects, unlike the research on partial  identification  (see e.g. 
 \cite{manski1990-AERpp}, \cite{manski1997} and \cite{manski-pepper-2000}).  
\cite{Chesher:03},
 \cite{HoderleinMammen2007}, \cite{FHMV2008}, \cite{imbens2009identification}, 
\cite{DF2015}, and \cite{Torgovitsky2015}   study 
models with continuous endogenous regressors.
Each of these papers develops identification results for various parameters of interest.
Our paper complements this literature by considering multivalued (but not continuous) treatments with more general types of selection mechanisms.

\citet[Appendix B]{HV2007-handbook} and \cite{HUV2008} and more recently \cite{heckmanpinto-pdt} and \cite{pinto-jmp} are  more closely related to our paper. But they focus on the selection induced by  multinomial discrete choice models, whereas our paper  allows for more general selection problems.

The  paper is organized as follows.  Section~\ref{sec:model} sets up our framework; it motivates our central assumptions by way of examples. We present and prove our identification results in Section~\ref{sec:general-results}.  Section~\ref{sec:apps} applies our results to  three important classes of applications, including the   models mentioned in this introduction. 
We relate our contributions to the literature in Section~\ref{sec:literature}.
Finally,  Section~\ref{sec:proof-iden} gives the proof of the main theorem.
  Some further results and  details of the omitted proofs are collected in  Online  Appendices.

\section{The Model and our Assumptions}
\label{sec:model}

We assume throughout that treatments take values in a  finite set of treatments  $\mathcal{K}$.  This set may be naturally ordered, as with different tax rates. But it may not be, as when welfare recipients enroll in  different training schemes for instance; this makes no difference to our results. We assume that treatments are exclusive. This involves no loss of generality as treatment values could easily be redefined otherwise. We denote 
$K=\abs{\mathcal{K}}$ the number of treatments, and we map the set $\mathcal{K}$ into
 $\{0,\ldots,K-1\}$ for notational convenience.

We denote $\{Y_k: k \in \mathcal{K} \}$ the potential outcomes.  Let  $D_k$ be 1 if the $k$ treatment is realized and
0 otherwise. The observed outcome and treatment  are $Y := \sum_{k \in \mathcal{K}} Y_k D_k$ and $D := \sum_{k \in \mathcal{K}} k D_k$, respectively.

In addition to the covariates $\bm{X}$, observed treatment $D$ and outcomes $Y$, the data contain a   random vector $\bm{Z}$ that will serve as instruments.
We  always condition on the value of $\bm{X}$ in our analysis of identification, and thus  suppress it from the notation. 
Observed data consist of a  sample 
$\{(Y_i,D_i,\bm{Z}_i): i=1,\ldots, N\}$ 
of $(Y,D,\bm{Z})$, where $N$ is the sample size.  We  denote the generalized propensity scores by
 $P_k(\bm{Z}) := \Pr(D = k\vert\bm{Z});$ they are directly identified from the data. Our models of treatment assignment  rely on functions of the instruments $Q_j\parZ$ that are a priori unknown to the econometrician and will need to be identified. We  also introduce  random vectors $\bm{V}$ to represent unobserved heterogeneity.

	Let $G$ denote a function defined on the  support $\mathcal{Y}$ of $Y$,  which can be discrete, continuous, or multidimensional. We focus on identification of the conditional counterfactual expectations $E\left(G(Y_k)\vert \bm{V}=\bm{v}\right)$ and on measures of treatment effects that can be derived from them.
	For example,  a possible object of interest is the marginal treatment effect (MTE), defined as $E\left(Y_k-Y_l\vert \bm{V}=\bm{v}\right)$. This is similar to the MTE in the binary treatment model, in that it conditions on the value of unobserved heterogeneity in treatment. One important difference is that the link between the unobserved heterogeneity vector $\bm{V}$ and the generalized propensity scores $\Pr(D=k\vert \bm{Z})$ is now more indirect. 
	
	Aggregating up would give the 
	mean of the counterfactual outcome $G(Y_k)$ (conditional on the omitted covariates $\bm{X}$). Once we identify $E G(Y_k)$ for each $k$, we also identify the average treatment effect  $E(G(Y_k) -G(Y_j))$ between any two treatments $k$ and $j$. 
	Alternatively, if we let 
	$G(Y_k) = \uniset(Y_k \leq y)$ for some $y$, where $\uniset(\cdot)$ is the usual indicator function, then the object of interest is the marginal distribution of  $Y_k$. This leads to the identification of quantile treatment effects.

One of our aims is to relax the usual monotonicity assumption that underlies the LATE and LIV estimators. Consider the following, simple example where $K=3$, and treatment assignment is driven by a pair of random variables $V_1$ and $V_2$ whose marginal distributions are normalized to be $U[0,1].$

\begin{example}[Selection with Two-Way Flows]\label{case:example1}
Assume that there are two thresholds $Q_1(\bm{Z})$ and $Q_2(\bm{Z})$ such that
\begin{itemize}
	\item $D=0$  iff $V_1<Q_1(\bm{Z})$ and $V_2<Q_2(\bm{Z})$,
	\item $D=1$ iff $V_1>Q_1(\bm{Z})$ and $V_2>Q_2(\bm{Z})$,
	\item $D=2$  iff $(V_1-Q_1(\bm{Z}))$ and $(V_2-Q_2(\bm{Z}))$ have opposite signs.
\end{itemize}
We could interpret $Q_1$ and $Q_2$ as minimum  grades or scores in a two-part exam or  an eligibility test based on two criteria: failing both parts/criteria  assigns you to $D=0$, passing both to $D=1$, and failing only one to $D=2$. 

If $F$ is the joint cdf of $(V_1,V_2)$, it follows that the generalized propensity scores are
\label{page:defEx1}
\begin{align}\label{two-way-flow-propensity-score}
\begin{split}
P_0\parZ &=F\left(Q_1\parZ,Q_2\parZ\right), \\
P_1\parZ &=1-Q_1\parZ-Q_2\parZ+F\left(Q_1\parZ,Q_2\parZ\right), \\
P_2\parZ &=	Q_1\parZ+Q_2\parZ-2F\left(Q_1\parZ,Q_2\parZ\right).
\end{split}
\end{align}
Take a  change in the values of the  instruments that increases both $Q_1\parZ$ and $Q_2\parZ$: both criteria, or both parts of the exam, become more demanding.
 Figure~\ref{fig:ex1} plots this change in $(V_1,V_2)$ space. 
 The black square represents the initial marginal observation, with $V_1=Q_1\parZ$ and $V_2=Q_2\parZ$; and the red circle at the other end of the arrow is the new marginal observation. In both cases, the quadrants delimited by the axes that intersect at the marginal observation define treatment cells. 
Observations in region (A)  move from $D=1$ to $D=2$, those in region (B)  move from $D=1$ to $D=0$, and those in regions (C)  move from $D=2$ to $D=0.$ This violates monotonicity, and even the weaker assumption that generalized propensity scores are monotonic in the instruments.   Note also that  observations  in region (C) leave $D=2$, while those in region (A) move into  $D=2$:  there are {\em two-way flows\/} in and out of $D=2$. Moreover, it is easy to see that {\em any\/} change in the thresholds creates such two-way flows; Figure~\ref{fig:ex1bis}  illustrates it for changes in opposite directions, with observations  in region (E) moving from $D=0$ to $D=2$,  observations (F) moving from $D=2$ to $D=1$, observations (G) moving from $D=1$ to $D=2$,
and observations (H)  moving from $D=2$ to $D=0$.

Therefore this model   violates the  weaker requirement of unordered monotonicity of \cite{heckmanpinto-pdt}, which we describe in Section~\ref{sub:unordered_monotonicity}---unless we are only interested in treatment values 0 and 1. $\qed$
\end{example}

\begin{figure}[h]
\begin{tikzpicture} 
\draw[->,>=latex] (-5,0)  -- (5,0) node[below] {$V_1$};
\draw[->,>=latex] (0,-5)  -- (0,5) node[above] {$V_2$};
\draw[dashed] (-5,-3) -- (5,-3);
\draw[dashed] (-2,-5) -- (-2,5);
\draw[dashed,red] (-5,-1) -- (5,-1);
\draw[dashed,red] (-1,-5) -- (-1,5);
\draw[->,>=latex] (-1.8,-2.8) -- (-1.2,-1.2);
	\draw[fill=black] (-2.1,-3.1) rectangle(-1.9,-2.9);
	\draw[fill=red] (-1,-1) circle(0.1);
\node[blue] at (-4,-4) {$D=0$};
\node[blue] at (1,1) {$D=1$};
\node[blue] at (-4,1) {$D=2$};
\node[blue] at (1,-4) {$D=2$};
\node[black] at (-1.6,-1.4) {(B)};
\node[black] at (2,-1.5) {(A)};
\node[black] at (-1.5,2) {(A)};
\node[black] at (-0.5,-1.5) {(A)};
\node[black] at (-1.5,-0.5) {(A)};
\node[black] at (-4,-1.5) {(C)};
\node[black] at (-1.5,-4) {(C)};
\end{tikzpicture} 
\caption{Example 1}
\label{fig:ex1}
\end{figure}
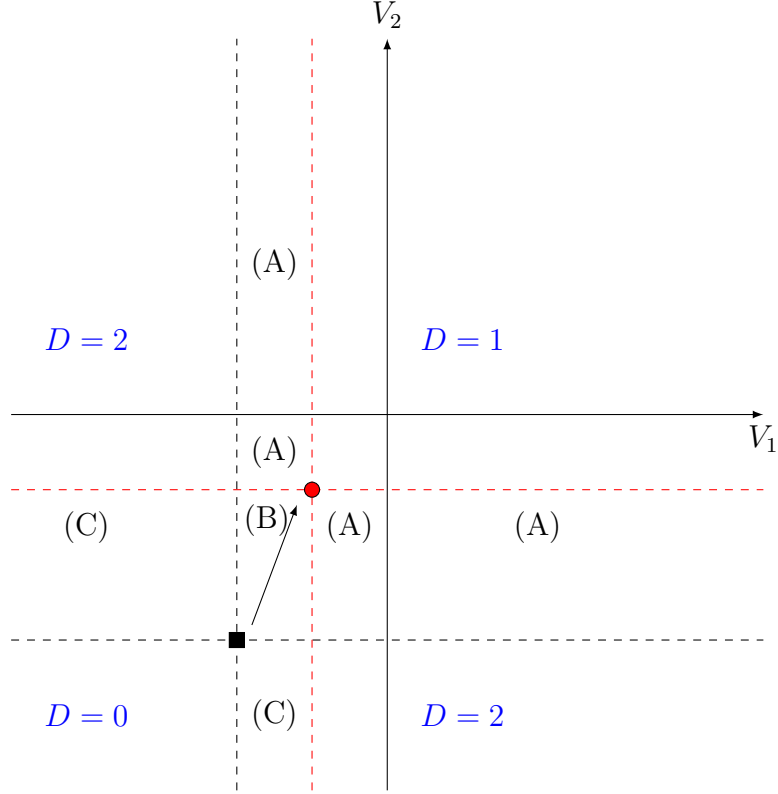

\begin{figure}[h]
\begin{tikzpicture} 
\draw[->,>=latex] (-5,0)  -- (5,0) node[below] {$V_1$};
\draw[->,>=latex] (0,-5)  -- (0,5) node[above] {$V_2$};
\draw[dashed, red] (-5,-3) -- (5,-3);
\draw[dashed] (-2,-5) -- (-2,5);
\draw[dashed] (-5,-1) -- (5,-1);
\draw[dashed,red] (-1,-5) -- (-1,5);
\draw[->,>=latex] (-1.8,-1.2) -- (-1.2,-2.8);
	\draw[fill=black] (-2.1,-1.1) rectangle(-1.9,-0.9);
	\draw[fill=red] (-1,-3) circle(0.1);
\node[blue] at (-4,-4) {$D=0$};
\node[blue] at (1,1) {$D=1$};
\node[blue] at (-4,1) {$D=2$};
\node[blue] at (1,-4) {$D=2$};
\node[black] at (2,-1.5) {(F)};
\node[black] at (-4,-1.5) {(E)};
\node[black] at (-1.5,2) {(G)};
\node[black] at (-0.5,-1.5) {(F)};
\node[black] at (-1.5,-0.5) {(G)};
\node[black] at (-1.5,-4) {(H)};
\end{tikzpicture} 
\caption{Example 1 (continued)}
\label{fig:ex1bis}
\end{figure}
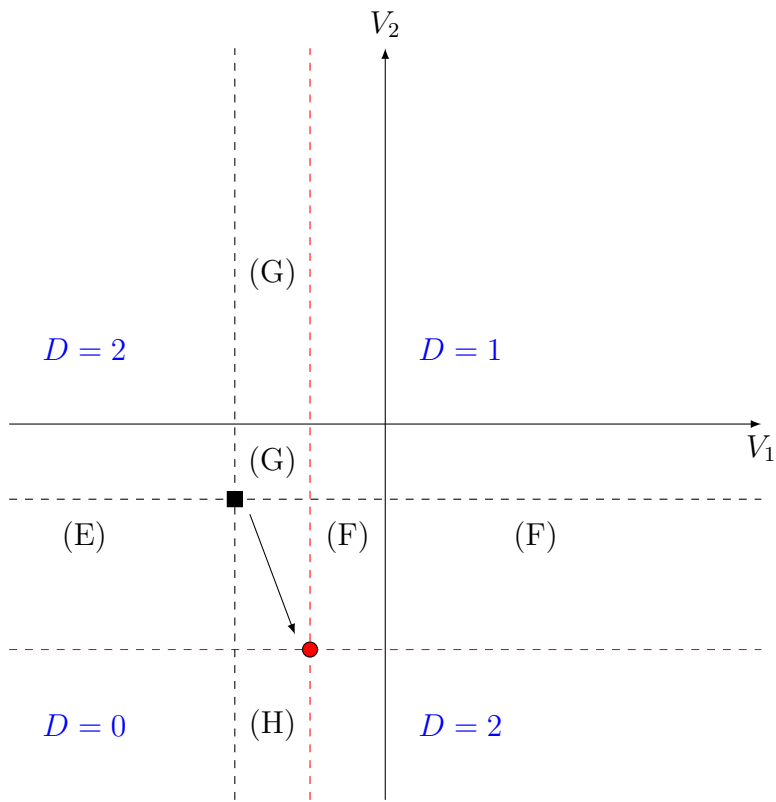

\medskip

To take a slightly more complicated example, consider the following entry game.

\begin{example}[Entry Game]\label{case:example2}
Two firms $j=1,2$ are considering entry into a new market. Firm $j$ has profit $\pi_j^m$ if it becomes a monopoly, and $\pi^d_j<\pi^m_j$ if both firms enter. The static Nash equilibria are simple:

\begin{itemize}
	\item if for both firms $\pi_j^m<0$, then no firm enters;
		\item if  $\pi_j^m>0$ and $\pi_k^m<0$, then only firm $j$ enters;
	\item if  for both firms $\pi_j^d>0$, then both firms enter;
	\item if  $\pi_j^d>0$ and $\pi_k^d<0$, then only firm $j$  enters;
	\item if  $\pi_j^m>0>\pi_j^d$ for both firms, then there are two symmetric equilibria, with only one firm operating.
\end{itemize}
Now let $\pi_j^m=V_j-Q_j(\bm{Z})$ and $\pi_j^d=\bar{V}_j-\bar{Q}_j(\bm{Z})$, and suppose we only observe the number $D=0,1,2$ of entrants. Then
\begin{itemize}
	\item $D=0$ iff $V_1<Q_1(\bm{Z})$ and $V_2<Q_2(\bm{Z})$
		\item $D=2$ iff $\bar{V}_1>\bar{Q}_1(\bm{Z})$ and $\bar{V}_2>\bar{Q}_2(\bm{Z})$
		\item $D=1$ otherwise. 
\end{itemize}
This is very similar to the structure of Example~\ref{case:example1}; in fact it coincides with it in the degenerate case when for each firm, $\pi^j_m$ and $\pi^j_d$ have the same sign with probability one\footnote{If the econometrician observes the identity of the entrants and not only their numbers, we face the usual partial identification problem generated by the existence of multiple equilibria \citep[see e.g.][]{Tamer2003}. If equilibrium selection is modeled as an additional threshold-crossing rule, then our approach actually encompasses this case. We refer the reader to Online Appendix \ref{sec:the_entry_game}, where we explain this in more detail.}. $\qed$
\end{example}


\subsection{The Selection Mechanism} 
\label{sub:selmech}

These two examples motivate the weak assumption we impose on the  underlying selection mechanism. 
In the following we use $\bm{J}$ to denote the set $\{1,\ldots,J\}.$

\begin{assumption}[Selection Mechanism]\label{assumption:selection}
There exist a finite number $J$, a vector of unobserved random variables $\bm{V} := \{V_j:  j \in \bm{J} \}$, and a vector of  
known 
functions 
$\{\bm{Q}_j\parZ:  j \in \bm{J} \}$ such that any of the following three  equivalent statements holds:
\begin{itemize}
\item[(i)]  the treatment variable $D$ is measurable with respect to the $\sigma$-field generated by the events 
\[
E_j\left(\bm{V},\bm{Q}(\bm{Z})\right):=\left\{V_j <Q_j(\bm{Z})\right\} \; \mbox{for} \;   j \in \bm{J};
\]
\item[(ii)] each event   $\{D=k\}=\{D_k=1\}$ is a member of this $\sigma$-field;
\item[(iii)] for each $k$, there exists a function $d_k$ that is measurable with respect to this $\sigma$-field such that $D_k=d_k(\bm{V},\bm{Q}(\bm{Z}))$.
\end{itemize}
Moreover, every treatment value $k$ has positive probability.
\end{assumption}

The threshold conditions in Assumption~\ref{assumption:selection} have the ``rectangular'' form $V_j<Q_j(\bm{Z})$. 
Appendix \ref{sec:nonrect} discusses a more general form of linear inequalities $\bm{\beta}_j\cdot \bm{V} <Q_j\parZ.$
Note that  the fact that every observation belongs to one and only one treatment group  imposes further constraints. We defer discussion of these constraints to section~\ref{sec:apps}, where we show how they can be used for overidentification tests.

  In this notation, the validity of the instruments translates into:

 \begin{assumption}[Conditional Independence of Instruments]
 	\label{ass:condindZ}
$Y_k$ and $\bm{V}$ are   jointly independent of $\bm{Z}$ for each $k=0,\ldots,K-1$.
 \end{assumption}

 To describe the class of selection mechanisms defined in Assumption \ref{assumption:selection} more concretely, we focus on a treatment value $k$. We  define   $S_j(\bm{V},\bm{Q}(\bm{Z})):=\uniset(V_j <Q_j(\bm{Z}))$ for 
  $j=1,\ldots,J$. The $\sigma$-field generated by  $\{E_j(\bm{V},\bm{Q}(\bm{Z})): j=1,\ldots,J\}$  is obtained by taking unions,  intersections, and complements of these $E_j$ sets. These three operations correspond to  taking sums, products, and differences of their indicator  functions $S_j$.
  Therefore the function $d_k$ referred to in Assumption~\ref{assumption:selection}.(iii) can be written as an algebraic sum of products of the $S_j$ indicator functions. Let $\mathcal{L}$ denote  the set of  all  subsets  $l=\left\{l_1,\ldots,l_{\abs{l}}\right\}$ of $\bm{J}$.  Then 
\begin{equation}
	\label{eq:dkpol}
d_k(\bm{V},\bm{Q}(\bm{Z}))=\sum_{l \in \mathcal{L}}  c_l^k \prod_{j\in l} S_{j}(\bm{V},\bm{Q}(\bm{Z})) =\sum_{l \in \mathcal{L}}  c_l^k \prod_{m=1}^{\abs{l}} S_{l_m}(\bm{V},\bm{Q}(\bm{Z}))
\end{equation}
where   the $c^k_l$ are algebraic integers. Moreover, this
 decomposition is unique.

Since $d_k(\bm{V},\bm{Q}(\bm{Z}))$ depends on 
	$\bm{V}$ and $\bm{Q}(\bm{Z})$ only through $\bm{S} := 
	\{S_{j}(\bm{V},\bm{Q}(\bm{Z})): j \in \bm{J} \}$, it will sometimes be convenient to express $d_k$ as a function of $\bm{S}$, which we denote $\mathcal{D}_k(\bm{S}).$
	For example, if $J=2$,  we have $\mathcal{D}_k(\bm{S}) = c^k_{\emptyset}+c^k_{\{1\}} S_1+ c^k_{\{2\}} S_2 + c^k_{\{1,2\}} S_1 S_2$ for some algebraic integers $c^k_{\emptyset},c^k_{\{1\}}, c^k_{\{2\}},$ and $c^k_{\{1,2\}}$.

To illustrate this, let us return to Example \ref{case:example1},  with $J=2$ and $K=3$.
For $k=0$, the selection mechanism is described by the intersection $E_1\cap E_2$, whose indicator function is 
$\mathcal{D}_0(\bm{S}) = S_1 S_2$.
Similarly,  for $k=1$ we find $\mathcal{D}_1(\bm{S})=(1-S_1) (1-S_2)$. Finally,   for $k=2$ we have 
 \[
 \mathcal{D}_2(\bm{S})=S_1(1-S_2)+(1-S_1)S_2=S_1+S_2-2S_1 S_2.
 \]

 It is useful to think of the products in~\eqref{eq:dkpol}  as  alternatives in a discrete choice model. For instance, $(1-S_1)S_2$ could be interpreted as ``item 1'' having negative value and ``item 2'' having positive value. In Example~\ref{case:example1},  $D_2=1$ informs us that the values of item 1 and of item 2 have opposite signs.  In essence, we are dealing with discrete choice models with only partially observed choices. This analogy will prove useful.

\subsection{Indices and Degrees} 
\label{sub:indices_and_degrees}

The term $l=\left\{1,\ldots,J\right\}=\bm{J}$, which corresponds to the product of all $J$ indicator functions $S_j$ in~\eqref{eq:dkpol},  plays an important role in our analysis. We will call its $c_l$ the {\em index of the treatment}.

\medskip

\begin{definition}\label{def:index}
	Take a treatment value $k$ in a treatment model with $J$ thresholds. We call the coefficient $c^k_{\bm{J}}$ in~\eqref{eq:dkpol} the index of treatment $k$.
\end{definition}

\medskip

In Example~\ref{case:example1}, the highest order term has a coefficient $c^k_{\{1,2\}}=-1$. With $J=2$ as in Example~\ref{case:example1},  the only  treatments with  a zero  index are those which  depend on only one threshold: e.g. $\uniset(V_1<Q_1)$. But with three or more thresholds ($J>2$), it is 
not hard to generate cases in which a treatment value $k$ depends on all $J$ thresholds and still has zero index, as shown in Example~\ref{case:example3}.

\begin{example}[Zero Index]\label{case:example3}
Assume  that $J=K=3$ 
 and take treatment $0$ such that 
\begin{align*}
D_0 &=\uniset(V_1<Q_1\parZ, V_2<Q_2\parZ, V_3<Q_3\parZ)\\
&+	\uniset(V_1>Q_1\parZ, V_2>Q_2\parZ, V_3>Q_3\parZ).
\end{align*}
Then the indicator function for $\{D_0=1\}$ is
\[
d_0=S_1S_2S_3+(1-S_1)(1-S_2)(1-S_3) =
1-S_1-S_2-S_3+S_1 S_2+S_1S_3+S_2S_3,
\]
which has no degree three term. $\qed$
\end{example}

When the  index is zero as in Example~\ref{case:example3}, the indicator function of the corresponding treatment $k$ has degree strictly smaller than $J$. Since Assumption~\ref{assumption:selection} rules out the uninteresting cases when treatment $k$ occurs with  probability zero or one, its indicator function cannot be constant; and its leading terms have degree $m\geq 1$.  We call $m$ the {\em degree\/} of treatment $k$. In Example~\ref{case:example3}, treatment value 0 has index 0 and degree 2.

The following lemma summarizes the discussion in Sections \ref{sub:selmech} and \ref{sub:indices_and_degrees}:
 
 \begin{lemma}\label{lem:pikn}
 	Under Assumption~\ref{assumption:selection}, for each $k\in\mathcal{K}$ there exists  a  unique family of algebraic integers $(c^k_l)$  such that 
	\[
d_k(\bm{V},\bm{Q}(\bm{Z})) =\sum_{l \in \mathcal{L}}  c^k_l \prod_{j\in l} S_{j}(\bm{V},\bm{Q}(\bm{Z}))
\]
	  where 
$\mathcal{L}$ is  the set of  all  subsets  $l=\left\{l_1,\ldots,l_{\abs{l}}\right\}$ of $\bm{J}$.
	  
 	The leading terms of the multivariate polynomial $\mathcal{D}_k(\bm{S})$ 	 have degree $1\leq m\leq J$, which we also call the {\em degree\/} of treatment $k$.
	 
	\begin{itemize}
		\item If $m=J$, then the leading term  in $\mathcal{D}_k(\bm{S})$ is 
	 \[
	   c^k_{\bm{J}} \prod_{j=1}^J S_j.
	 \]
	 \item if $m<J$, then $c^k_{\bm{J}}=0$.
	\end{itemize} 
	We call $c^k_{\bm{J}}$ the \emph{index} of treatment $k$.
 \end{lemma}

\section{Identification Results}\label{sec:general-results}

In this section we fix $\bm{x}$ in the support of $\bm{X}$ and we suppress it from the notation.
All the results obtained below are local to this choice of $\bm{x}$.
Global (unconditional) identification results follow immediately if our assumptions hold for almost every $\bm{x}$ in the support of $\bm{X}$.

 We only  treat  the non-zero index in the text. We make this  explicit in the following assumption.

 \begin{assumption}[Non-zero index]\label{ass:index}
 The index $c^k_{\bm{J}}$ defined in Lemma \ref{lem:pikn} is nonzero. 
 \end{assumption}

 We   analyze zero-index treatments in  Appendix~\ref{sub:index-zero-case}.
 
 \medskip

We require that $\bm{V}$ have full support:

\begin{assumption}[Continuously Distributed Unobserved Heterogeneity in the Selection Mechanism]
	\label{ass:continuous}
The joint distribution of $\bm{V}$   is absolutely continuous with respect to the Lebesgue measure on $\mathbb{R}^{J}$ and its support is $[0,1]^{J}$. 
\end{assumption}




Note that when $J=1$, Assumptions \ref{assumption:selection}  and \ref{ass:continuous} define the usual threshold-crossing model that underlies the LATE and LIV approaches. However, our assumptions  allow for a much richer class of selection mechanisms when $J>1$.  Our Example~\ref{case:example1} illustrates that  our ``multiple thresholds model'' does not impose any multidimensional extension of the monotonicity condition that is implicit with a single threshold model. Even when $K=2$, so that treatment is binary, $J$ could be larger than one. This would  allow for flexible treatment assignment: just modify Example~\ref{case:example1} to obtain the double hurdle model
\[
D=\uniset\left(V_1<Q_1(\bm{Z}) \mbox{ and }  V_2<Q_2(\bm{Z})\right).
\]

Let  
$f_{\bm{V}}(\bm{v})$ denote the joint density function of $\bm{V}$ at $\bm{v} \in [0,1]^{J}$. 
Our identification argument relies on continuous instruments that generate enough variation in the thresholds. This motivates the following three assumptions.

For any function $\psi$ of $\bm{q}$, define ``local equicontinuity at $\bm{\overline{q}}$'' by the following property:
for any subset $I\subset \bm{J}$, the family of functions $\bm{q}_I \mapsto \psi(\bm{q}_I, \bm{q}_{-I})$ indexed by $\bm{q}_{-I} \in [0,1]^{\abs{\bm{J}-I}}$ is equicontinuous in a neighborhood of $\bm{\overline{q}}_I.$
 
\begin{assumption}[Local equicontinuity at $\bm{q}$]\label{ass:Q-continuity}
The functions $\bm{v} \mapsto f_{\bm{V}}(\bm{v})$ and
 $\bm{v} \mapsto  E\left(G(Y_k)\vert \bm{V}=\bm{v}\right)$
 are locally equicontinuous at $\bm{v}=\bm{q}$.
\end{assumption}

 Assumption \ref{ass:Q-continuity}    allows us to differentiate the relevant expectation terms.  It is fairly weak:  Lipschitz-continuity for instance implies local equicontinuity\footnote{It would be easy to adapt our results to cases where, for instance, $\bm{Q}$ has discontinuities. We do not pursue it in this paper.}.

\begin{definition}
	Let $\mathcal{Z}$ denote the support of $\bm{Z}$; and $\mathcal{Q}=\bm{Q}(\mathcal{Z})$ the range of variation of $\bm{Q}\parZ$.
\end{definition}

 The next  two assumptions apply to the functions $\bm{Q}(\bm{Z})$  and in particular to their range of variation over the support $\mathcal{Z}$ of $\bm{Z}$.  The functions $\bm{Q}$ are unknown in most cases, and need to be identified;  in this part of the paper we assume that they are known. We will return to identification of the $\bm{Q}$ functions in Section~\ref{sec:identQ}.

\begin{assumption}[Open Range at $\bm{q}$]\label{ass:open} 
The point $\bm{q}$ belongs to the interior of  the range of variation of the thresholds $\mathcal{Q}$.
\end{assumption}

Assumption~\ref{ass:open} ensures that we can generate any small variation in $\bm{Q}(\bm{Z})$  around $\bm{q}$ by varying the instruments around $\bm{z}.$ 
This makes 
the instruments strong enough to deal with multidimensional unobserved heterogeneity $\bm{V}$.  

With $J$ thresholds, Assumption~\ref{ass:open} requires that $\mathcal{Q}$ contains a $J$-dimensional neighborhood of $\bm{q}$. This in turn can only happen (given Assumption~\ref{ass:Q-continuity}) if the range of variation of the instruments $\mathcal{Z}$ contains an open subset of $\Real^J.$ 
Having $J$-dimensional continuous variation in the instruments is crucial to our approach.   

\medskip

For some corollaries, we  use a global version of  Assumptions~\ref{ass:Q-continuity} and \ref{ass:open}.  To state it formally, we need one last definition.

\begin{definition}
	Let $\mathcal{\tilde{Q}}\subset \mathcal{Q}$ denote the set of values $\bm{q}$ where Assumptions~\ref{ass:Q-continuity} and~\ref{ass:open} both hold.
\end{definition}

\begin{assumption}[Global Condition]	\label{ass:global}
$\mathcal{\tilde{Q}}$ contains $(0,1)^J$.
\end{assumption}
 
Assumption \ref{ass:global}  requires  both that the variation in the instruments generate all possible values of the $J$ thresholds and that Assumptions~\ref{ass:Q-continuity} and~\ref{ass:open} hold everywhere. 
We do not need this rather stringent assumption to identify the marginal treatment effects;  but it is useful to derive various  parameters of interest that aggregate the marginal treatment effects.

\subsection{Identification with a Non-Zero Index} 
\label{sub:identification_with_a_non_zero_index}

 We are now ready to prove identification of $E\left(G(Y_k)\vert \bm{V}=\bm{q}\right)$ when treatment $k$ has a non-zero index.   In the following theorem,  for any real-valued function $\bm{q}\mapsto h(\bm{q})$, the notation
\[
T h(\bm{q})\equiv \frac{\partial^J h}{\prod_{j=1}^J\partial q_j}(\bm{q})
\]
refers to the $J$-order derivative that obtains by taking derivatives of the function $h$ at $\bm{q}$ in each direction of $\bm{J}$, when this derivative exists.

\begin{theorem}[Identification with a non-zero index]\label{thm:iden}
Let Assumptions \ref{assumption:selection}, \ref{ass:condindZ}, \ref{ass:index}, and \ref{ass:continuous}   hold. Fix a value  $\bm{q}$ where Assumptions~\ref{ass:Q-continuity} and \ref{ass:open} hold; that is, $\bm{q}\in\mathcal{\tilde{Q}}$.  Then	 the density  of $\bm{V}$ and  the conditional expectation of $G(Y_k)$ are given by\footnote{The proof of the theorem shows that these derivatives are well-defined.} 
\begin{align*}
f_{\bm{V}}(\bm{q}) 
&= \frac{1}{c^k_{\bm{J}}}
T \Pr(D=k\vert\bm{Q}\parZ =\bm{q}) \\[5mm]
 E[G(Y_k) \vert \bm{V} =\bm{q}] 
&= 
\frac{T E\left(G(Y)D_k\vert\bm{Q}\parZ = \bm{q}\right)}
{T \Pr(D=k\vert\bm{Q}\parZ = \bm{q})}.
\end{align*}
\end{theorem}

\begin{proof}[Proof of Theorem \ref{thm:iden}] 
See section~\ref{sec:proof-iden}.
\end{proof}


For two treatment values  $k$ and $\ell$, define the marginal treatment effect as 
\begin{align}\label{MTE}
\Delta_{\text{MTE}}^{(k, \ell)} (\bm{v}) &:= 
E[G(Y_k) \vert \bm{V} =\bm{v}] - E[G(Y_\ell) \vert \bm{V} =\bm{v}].
\end{align}
The MTE function $\bm{v} \mapsto \Delta_{\text{MTE}}^{(k, \ell)} (\bm{v})$ is the average treatment effect conditional on $\bm{V}=\bm{v}$. Since 
$\bm{V}$ is the vector of unobservables that determine the selection mechanism, the MTE function reveals how treatment effects vary with the unobservables governing selection.  As such, it captures the effect of selection and it allows the analyst to  simulate counterfactual policies.  It follows from Theorem \ref{thm:iden} that if $k$ and $\ell$ are two treatments to which all of our assumptions apply, then we can identify the  marginal treatment effect of moving between these two treatments, as well as the quantile version of this MTE.
We also identify the joint density function $\bm{v} \mapsto f_{\bm{V}} (\bm{v})$, which is an  object of interest since it describes the  dependence among elements of  $\bm{V}$. Appendix~\ref{sub:index-zero-case} extends Theorem~\ref{thm:iden} to zero-index treatment values; it shows that similar formul\ae\ identify marginal treatment effects averaged over the missing threshold rules.
 	
As in \cite{MIV2005}, we can identify various treatment effect parameters using Theorem  \ref{thm:iden}.
The following corollary shows that one can identify the average treatment effect (ATE), the average treatment effect on the treated (ATT), and the policy relevant treatment effect (PRTE) of \cite{heckman2001policy}. 
The PRTE  measures the 
 average effect of moving from a baseline policy to an alternative policy.
To define the PRTE, consider a class of policies that change $\bm{Q}$ but that do not affect $E[G(Y_k) \vert \bm{V} =\bm{v}]$. 
Let  $D_k^\ast$ and $Y^\ast$, respectively, denote the treatment choice indicator and the outcome under a new policy $\bm{Q}^\ast$. 
Define $D^\ast \equiv \sum_{k \in \mathcal{K}} k D_k^\ast$.

\begin{corollary}\label{iden:ATE-ATT}
If Assumption~\ref{ass:global} holds in addition to the conditions assumed in Theorem \ref{thm:iden}, then the average treatment effect (ATE) and the average treatment effect on the treated (ATT) are identified by 
\begin{align}
E [G(Y_k) - G(Y_\ell)] &=  \int   \Delta_{\text{MTE}}^{(k, \ell)} (\bm{v}) \omega_{\text{ATE}}(\bm{v}) d\bm{v}, \label{iden-eq-ATE} \\
 E[G(Y_k) - G(Y_\ell)|D = k]  &= \int  \Delta_{\text{MTE}}^{(k, \ell)} (\bm{v})  \omega^k_{\text{ATT}}(\bm{v})  d\bm{v},	\label{iden-eq-ATT}
\end{align}
where
\begin{align*}
\omega_{\text{ATE}}(\bm{v}) &:=   f_{\bm{V}}(\bm{v}), \\
\omega^k_{\text{ATT}}(\bm{v}) &:= 
\frac{\Pr \left[ d_k(\bm{v},\bm{Q}\parZ)=1 \vert \bm{V}=\bm{v} \right]  f_{\bm{V}}(\bm{v})}{\Pr(D=k)}.
\end{align*}
Furthermore, policy relevant treatment effects (PRTEs) are identified by 
\begin{align*}
E[G(Y^\ast)]  -  E[G(Y)] 
&=  \sum_{k \in \mathcal{K}} \int 
\Upsilon_{k}(\bm{v},\bm{Q^\ast},\bm{Q})
E[G(Y_k) \vert \bm{V} =\bm{v}] f_{\bm{V}}(\bm{v}) 
d\bm{v}, \\
E[D^\ast] - E[D]
&=  \sum_{k \in \mathcal{K}} k \int 
\Upsilon_{k}(\bm{v},\bm{Q^\ast},\bm{Q})
f_{\bm{V}}(\bm{v}) 
d\bm{v}, \\
E[D^\ast_k = 1] - E[D_k = 1]
&=   \int 
\Upsilon_{k}(\bm{v},\bm{Q^\ast},\bm{Q})
f_{\bm{V}}(\bm{v}) 
d\bm{v},
\end{align*}
where
\begin{align*}
\Upsilon_{k}(\bm{v},\bm{Q^\ast},\bm{Q}) := \Pr[ d_k(\bm{v},\bm{Q^\ast}\parZ) =1 \vert \bm{V}=\bm{v}]  
- \Pr[ d_k(\bm{v},\bm{Q}\parZ) =1 \vert \bm{V}=\bm{v}].
\end{align*}
\end{corollary}

\begin{proof}[Proof of Corollary \ref{iden:ATE-ATT}]
See Appendix~\ref{appx:proof-iden-ATE-ATT}.
\end{proof}

In many applications, the range of variation of the thresholds may be limited so that
  Assumption~\ref{ass:global}  will not hold. 
However, it is still possible to construct bounds for the ATE, ATT and PRTE if $G(Y_k)$ is bounded. 
For example, consider the ATE with $G(Y_k) = \uniset(Y_k \leq y)$. As shown in the proof of Theorem \ref{iden:ATE-ATT}, we can point-identify 
$E[G(Y_k) \vert \bm{V} =\bm{q}]f_{\bm{V}}(\bm{q}) $  
 by $\left(c^k_{\bm{J}}\right)^{-1} {T E\left(G(Y)D_k\vert\bm{Q}\parZ = \bm{q}\right)}$ for each  $\bm{q} \in \mathcal{\tilde{Q}}$. 
 In addition, we know that $G(Y_k)$ lies between 0 and 1. As in \cite{manski1990-AERpp} and \cite{NBERt0252}, using this fact  we can bound 
$EG(Y_k)$ within the  interval defined by
\[
\frac{1}{c^k_{\bm{J}}} \int_{\bm{q} \in \mathcal{S}_{\bm{Q}(\bm{Z})}} {T E\left(G(Y)D_k\vert\bm{Q}\parZ = \bm{q}\right)}  d\bm{q} 
\]
and
\[
\frac{1}{c^k_{\bm{J}}} \int_{\bm{q} \in \mathcal{S}_{\bm{Q}(\bm{Z})}} {T E\left(G(Y)D_k\vert\bm{Q}\parZ = \bm{q}\right)}  d\bm{q}
+ 1-\Pr\left(\bm{Q}(\bm{Z})\in \mathcal{Q}\right);
\]
and without further information, these bounds are sharp.

Finally, the analyst may only have  discrete-valued instruments. A recent literature on the MTE focuses on this case (with binary treatment); it relies on assumed restrictions on the shape of the MTE function (see for example \cite{BeyondLATE}, \cite{NBERw22363} and \cite{MST2016}).  In future work it would be  interesting to consider relaxing these restrictions within our framework.



\subsection{Identification of $\bf{Q}$}
\label{sec:identQ}

 So far we assumed that the functions $\{\bm{Q}_j \parZ: j=1,\ldots,J\}$  were known (see Assumption \ref{assumption:selection}). In practice we   often need to 
identify them from the data before
applying Theorems~\ref{thm:iden} or~\ref{thm:idenzero}. The most natural  way to do so starts from the generalized propensity scores
 $\{P_k\parZ: k=0,\ldots,K-1\}$, which are  identified as the conditional probabilities of treatment\footnote{It would also be possible to seek identification jointly from the generalized propensity scores and from the cross-derivatives that appear in Theorems~\ref{thm:iden} or~\ref{thm:idenzero}, especially when they are over-identified. We do not pursue this here.}.

	First note that by definition (and by Assumption~\ref{ass:condindZ}),
	\begin{align*}
	P_k(\bm{z})&=\Pr(D=k\vert\bm{Z}=\bm{z}) 
	\\
	&=\int \uniset\left(d_k\left(\bm{v},
	\bm{Q}(\bm{z})\right)=1\right)f_{\bm{V}}(\bm{v})d\bm{v}.
	\end{align*}
Note that this is a $J$-index model.
\cite{ichilee:multindex} consider identification of multiple index models when the indices are specified parametrically.
\cite{matzkin1993-JoE, matzkin2007-advances} obtains nonparametric identification results for discrete choice models\footnote{See  \citet[Appendix B]{HV2007-handbook} for an application to treatment models.};  
but her  results only  apply to a subset of  the types of selection mechanisms we  consider (discrete choice models when all choices are observed). Section~\ref{sec:apps} discusses identification of the $\bm{Q}$'s through the lens of  several models.

\section{Applications}

\label{sec:apps}

Our framework covers a wide variety of commonly used models. 
 For simplicity, we only illustrate its usefulness on two-threshold selection models in this section.
These models generate different selection patterns. Not surprisingly, the identification conditions require
somewhat stronger instruments as the number of treatment values---the information available to the analyst---decreases.

\subsection{Selection with Two-Way Flows} 
\label{sub:two_way_flows}

Let us return to Example~\ref{case:example1}, in which 
\begin{itemize}
	\item $D=0$  iff $V_1<Q_1(\bm{Z})$ and $V_2<Q_2(\bm{Z})$,
	\item $D=1$ iff $V_1>Q_1(\bm{Z})$ and $V_2>Q_2(\bm{Z})$,
	\item $D=2$  iff $(V_1-Q_1(\bm{Z}))$ and $(V_2-Q_2(\bm{Z}))$ have opposite signs.
\end{itemize}

It is useful to start with some exclusion restrictions that help us identify $Q_1(\mathbf{Z})$ and $Q_2(\mathbf{Z})$ separately from the generalized propensity scores given in \eqref{two-way-flow-propensity-score}.
Assume that

\begin{assumption}[Two Continuous Instruments with Exclusion Restrictions]
	\label{ass:Q:iden-two-way-flows}
	\mbox{}
\begin{enumerate}
	\item The density of $(V_1,V_2)$ is continuous on $[0,1]^2$,  with  marginal uniform distributions.
	\item  The instruments $\mathbf{Z} \equiv (Z_1,Z_2)$ consist of  two scalar random variables 
	whose joint
	 distribution is absolutely continuous with respect to the Lebesgue measure  on its support $\mathcal{Z}$.
	 \item 
	 $Q_1(\mathbf{Z})$  does not depend on $Z_2$, and it is  continuously differentiable with respect to $Z_1$.
	 \item 
	 $Q_2(\mathbf{Z})$  does not depend on $Z_1$, and it is  continuously differentiable with respect to $Z_2$.
\end{enumerate} 
\end{assumption}	

The first condition in Assumption \ref{ass:Q:iden-two-way-flows} is just a normalization of  the marginal distribution of each $V_j \in \bm{V}$.
 The crucial  part of Assumption \ref{ass:Q:iden-two-way-flows}
is in the exclusion restrictions:
$Z_1$ affects  $Q_1$ but not $Q_2$,  and  $Z_2$ affects  $Q_2$ but not $Q_1$.
For example, if  $Q_1$ and $Q_2$ represent minimum required grades on two parts of an exam, then $Z_1$  should affect only the  requirement on the first part, and 
$Z_2$ should only affect the second part.


\begin{theorem}[Identification of $Q_1$ and $Q_2$]\label{pro:Q-two-way-flows}
	Under Assumption~\ref{ass:Q:iden-two-way-flows}, 
	\begin{itemize}
		\item[(i)]
	the function
	\[
	P\parZ \equiv 2P_0\parZ+ P_1\parZ
	\]
	is additively separable in $Z_1$ and $Z_2$ on $\mathcal{Z}$. 
\item[(ii)]
	$Q_1$ and $Q_2$ are identified up to an additive constant. More precisely, take any $(z_1^0,z_2^0)\in \mathcal{Z}$. Then
	\begin{align*}
		Q_1(z_1) &= P(z_1,z_2^0)-P(z_1^0,z_2^0)+C_1^0\\
		Q_2(z_2) &= P(z_1^0,z_2)-C_1^0
	\end{align*}
where the constant $C_1^0$ must satisfy  the restrictions 
$\Pr(D = k) > 0$ for each $k=0,1,2$.\footnote{The precise form of these restrictions in terms of $C_1^0$ and $P(Z_1,Z_2)$ is given  in the proof of Theorem  \ref{pro:Q-two-way-flows}.}
\end{itemize}
\end{theorem}

\begin{proof}[Proof of Theorem \ref{pro:Q-two-way-flows}] 
See  Appendix~\ref{appendix:Q-two-way-flows}.
\end{proof}

  Suppose that the analyst has picked a point in the partially identified $(Q_1,Q_2)$ set.  Using these $Q_1$ and $Q_2$,   since the indices are all nonzero ($c^0_{\bm{J}}=c^1_{\bm{J}}=1$ and $c^2_{\bm{J}}=-2$) we apply Theorem~\ref{thm:iden} to identify the joint density by 
\begin{equation}
\label{eq:two-way-density}
f_{V_1,V_2}(q_1,q_2) = \frac{1}{c^k_{\bm{J}}}\frac{\partial^2 \Pr[D=k\vert Q_1(\bm{Z}) = q_1, Q_2(\bm{Z})=q_2]}{\partial q_1  \partial q_2},
\end{equation}
where $k=0,1,2$.

Note that $f_{V_1,V_2}(q_1,q_2)$ is overidentified;    checking equality between the right-hand sides of 
 \eqref{eq:two-way-density}  for $k=0,1,2$ provides a specification test\footnote{Since probabilities add up to one, only one of these equalities generates a specification test.}. Similar remarks apply to 
the conditional expectations  $E(Y_k \vert V_1=q_1, V_2=q_2)$; and  as
\begin{align*}
	E(Y_k \vert V_1=q_1, V_2=q_2) &=
	 \frac{\partial^2 E[ Y D_k\vert Q_1(\bm{Z}) = q_1, Q_2(\bm{Z})=q_2]/\partial q_1  \partial q_2 }
	{{\partial^2 \Pr[ D=k\vert Q_1(\bm{Z}) = q_1, Q_2(\bm{Z})=q_2]}/{\partial q_1  \partial q_2} }
 \end{align*}
for each $k=0,1,2$,
the identification of the marginal and average treatment effects follows immediately.

 In practice,   
$Q_1$ and $Q_2$ are only identified up to the (restricted) additive constant $C_1^0$ in Theorem \ref{pro:Q-two-way-flows}(ii).
As a consequence,  $f_{V_1,V_2}(q_1,q_2)$ and $E(Y_k \vert V_1=q_1, V_2=q_2)$ are only identified up to the corresponding location shift in $(q_1,q_2)$. However, it is easy to check that~\eqref{eq:two-way-density} still yields a usable specification test.

\subsection{The Double Hurdle Model} 
\label{sub:multiple_hurdles}

Let us now return to the double hurdle model of the introduction, where treatment is binary and  the selection mechanism is governed by 
\begin{align}\label{mh-example}
D=1 \mbox{ iff } V_1 < Q_1(\bm{Z}) \mbox{ and } V_2 < Q_2(\bm{Z}),
\end{align}
and $D= 0$ otherwise.

Both treatment values have  non-zero indices:   $c^1_{\bm{J}}=1$ and $c^0_{\bm{J}}=-1.$
 But identification of $Q_1$ and $Q_2$, which is a premise of Theorem~\ref{thm:iden}, is far from  straightforward. 
In fact, this case is more demanding than the selection model with two-way flows in 
Section \ref{sub:two_way_flows} since we only have two treatment values.
We observe  the  propensity score
\begin{align}\label{prob_d_1_double_hurdle}
\Pr(D=1\vert \bm{Z})=F_{V_1,V_2}\left(Q_1(\bm{Z}), Q_2(\bm{Z})\right),
\end{align}
 which  we denote $H(\bm{Z})$. This is a nonparametric double index model in which both the link function $F_{V_1,V_2}$ and the indices $Q_1$ and $Q_2$ are unknown; it is clearly underidentified without stronger restrictions. 
 \cite{matzkin1993-JoE, matzkin2007-advances} considers nonparametric identification and estimation of polychotomous choice models. 
Our multiple hurdle model has a similar but not identical structure.

In order to identify $\bm{Q}$, we assume  that there exist two instruments that are excluded from one of the thresholds. More precisely, let  the vector of instruments be $\bm{Z}=(Z_1,Z_2,\bm{Z}_{-12})$, with $Z_1$ and $Z_2$ scalar; we require that
\begin{itemize}
	\item $Q_1(\bm{Z})$ does not depend on $Z_2$, and 
	\item $Q_2(\bm{Z})$ does not depend on $Z_1$.
\end{itemize} 
To simplify notation, we  fix  the value of $\bm{Z}_{-12}$ and we denote $Q_1\parZ =G_1\left(Z_1\right)$ and $Q_2\parZ=G_2(Z_2)$, where $G_1$ and $G_2$ are two unknown functions.   Note that the propensity score becomes $H(Z_1,Z_2)=F_{V_1,V_2}(G_1(Z_1),G_2(Z_2))$.

	We  give two identification results under these exclusion restrictions. 
	We first build on \cite{lewbel2000} and on  \cite{matzkin1993-JoE, matzkin2007-advances}'s results to identify $\bm{Q}$ and rely on full support restrictions (conditional on the value of $\bm{Z}_{-12}$):

\begin{assumption}\label{ass:Q:globiden}
$Q_1=G_1(Z_1)$ and $Q_2=G_2(Z_2)$. Moreover,
		\begin{enumerate}
			\item The density of $(V_1,V_2)$ is continuous on $[0,1]^2$,  with  marginal uniform distributions.
			
		\item  $G_1$ and $G_2$ are strictly increasing $C^1$ functions  from possibly unbounded intervals $(a_1,b_1)$ and $(a_2,b_2)$  to $(0,1)$; that is, for every $t\in(0,1)$ there exist $z_1\in (a_1,b_1)$ and $z_2\in (a_2,b_2)$ such that $G_1(z_1)=G_2(z_2)=t$.
			\item $\mathcal{Z}$  is the rectangle $(a_1,b_1)\times (a_2,b_2)$.
		\end{enumerate}
\end{assumption}

\begin{theorem}\label{pro:globident2hurdle}
	Under Assumption~\ref{ass:Q:globiden}, the functions $F_{\bm{V}}, G_1$ and $G_2$ are identified from the propensity score ${\Pr(D=1\vert \bm{Z})}.$
\end{theorem}

\begin{proof}
See Appendix~\ref{sub:proof_of_theorem_ref_pro_globident2hurdle}.
\end{proof}

While Theorem \ref{pro:globident2hurdle} requires two continuous instruments that generate all possible values of the thresholds, various additional restrictions would relax this requirement. If for instance  $G_1$ and $G_2$ were linear, we would be back to the linear multiple index  model of \cite{ichilee:multindex}.

Theorem \ref{pro:locident2hurdle} provides a complementary result and is useful when the instruments have limited support. It relies
on a semiparametric restriction. Remember that we normalized  the marginal distributions of $V_1$  and $V_2$ to be uniform over $[0,1]$; we now assume that the codependence between $V_1$ and $V_2$ is described by a strict symmetric Archimedean copula\footnote{The class of Archimedean copulas include the Clayton, Frank, and Gumbel families among others \citep[see][ch. 4]{Nelsen:2006}.}:
\begin{equation}
	\label{eq:archi}
	F_{V_1,V_2}(v_1,v_2)=\phi^{-1}\left(\phi(v_1)+\phi(v_2)\right),
\end{equation}
 where $\phi$ belongs to the set $\Psi$ of  $C^2$, strictly decreasing, and convex functions from $[0,1]$ unto $[0,+\infty]$.

\begin{assumption}\label{ass:Q:lociden} $Q_1=G_1(Z_1)$ and $Q_2=G_2(Z_2)$. Moreover,
		\begin{enumerate}
			\item[(a)] 
			the propensity score and 
			the distribution of $(V_1,V_2)$ are described by \eqref{prob_d_1_double_hurdle} and \eqref{eq:archi} for unknown functions $\phi$, $G_1$, and $G_2$,
			\item[(b)] the interior of the support  $\mathcal{Z}$ of $(Z_1,Z_2)$    contains a connected set $\mathcal{N}$, and
			\item[(c)] $G_1$ and $G_2$ are $C^1$ functions over the projections of $\mathcal{N}$, with derivatives bounded away from zero.
\end{enumerate}
\end{assumption}

Let $H_k(z_1,z_2)$ denote the derivative of  the propensity score $H(z_1,z_2)$ with respect to its $k$th argument ($k = 1, 2$) and 
$H_{12}(z_1,z_2)$ the second-order cross derivative of $H(z_1,z_2)$.  Note that the scale of $\phi$ is not identifiable in view of \eqref{eq:archi}.
Furthermore, if $\phi$ is  specified nonparametrically as an element of $\Psi$, the location of $\phi$ is only identifiable when the argument of $\phi$ takes values close to~1:

\begin{theorem}\label{pro:locident2hurdle}
	Let Assumption~\ref{ass:Q:lociden} hold. Then
	\begin{enumerate}
			\item  Over $\mathcal{N}$, the ratio
			\[
		\frac{H_{12}}{H_1 H_2}(z_1,z_2)
			\]
			is non-negative and only depends on the value $h=H(z_1,z_2)$.
		\item The function $\phi$ is   identified up to scale and location in $\Psi$ on the image $H(\mathcal{N})=(\underline{h},\bar{h})\subset(0,1)$ of $\mathcal{N}$ under $H$, where  $\mathcal{N}$ is given in Assumption~\ref{ass:Q:lociden}(b).
		\item The scale parameter for $\phi$ is an arbitrary negative number,  which we normalize by imposing $\phi'(\bar{h}) = -1$; given this normalization, the location parameter for $\phi$ is   bounded by 
		\[
		0\leq \phi(\bar{h}) \leq 1-\bar{h}.
		\]
  If moreover $\sup_{\bm{z}\in \mathcal{N}} \Pr(D=1\vert \bm{Z}=\bm{z})=1$, then $\bar{h}=1$ and  $\phi$  is point-identified.
		\item For any admissible value of the location parameter of $\phi$, the functions $G_1$ and $G_2$ are  identified on $\mathcal{N}$
		up to a common constant $k$: 
 any other admissible $(\tilde{G}_1,\tilde{G}_2)$ must satisfy
	\begin{align*}
		\phi(\tilde{G}_1(z_1)) &= \phi(G_1(z_1))-k \\
		\phi(\tilde{G}_2(z_2)) &= \phi(G_2(z_2))+k
	\end{align*}
over  the projections of $\mathcal{N}$. The number $k$ is bounded above and below.
  If moreover $\sup_{\bm{z}\in \mathcal{N}} \Pr(D=1\vert \bm{Z}=\bm{z})=1$, then $G_1$ and $G_2$ are point-identified on the projections of $\mathcal{N}$.
\end{enumerate}
\end{theorem}

\begin{proof}
See Appendix~\ref{sub:proof_of_theorem_ref_pro_locident2hurdle}.
\end{proof}

\medskip

Our constructive identification starts by writing
\begin{align}\label{eq-copula-iden}
\frac{\phi^{\prime\prime}}{\phi^\prime}(h)
=-\frac{H_{12}}{H_1 H_2}(z_1,z_2)
\end{align}
for all $(z_1,z_2)$ such that $H(z_1,z_2)=h$.
Once $\phi$ is identified from \eqref{eq-copula-iden}, then we proceed to identify $G_1$ and $G_2$.
The function $G_1$, for instance, would be identified by
\[
\phi(G_1(z_1))=\phi(H(z_1,z^0_2))-\phi(g^0_2)
\]
for a fixed $z^0_2$ and a value $g^0_2$ of $G_2(z^0_2)$.

Given more a priori restrictions on the  function $\phi$, identification results can be sharper. The following example illustrates this point by taking a parametric family of $\phi$.

\begin{example}
Take the strict Clayton copula, which is generated by  $\phi(u)=(u^{-\theta}-1)/\theta$ for $\theta>0$. This yields
\[
H(z_1,z_2)=\left(G_1(z_1)^{-\theta}+G_2(z_2)^{-\theta}-1\right)^{-1/\theta}.
\]
In this example,   $\frac{\phi^{\prime\prime}}{\phi^\prime}(h)$ is simply $-(1+\theta)/h$. Therefore, it follows from \eqref{eq-copula-iden} that $\theta$ can be identified in closed form as
\begin{align}\label{theta-clayton}
\theta=h\frac{H_{12}}{H_1 H_2}(z_1,z_2)-1
\end{align}
for all $(h,z_1,z_2)$ such that $H(z_1,z_2)=h$.  Note that the scale and location of $\phi$ are point-identified, given the parametric restriction.

Conversely, the constancy of the right-hand side of \eqref{theta-clayton} characterizes a Clayton copula. 
To identify $G_1$ and $G_2$, note that
\begin{align*}
G_1(z_1)^{-\theta} + G_2(z_2)^{-\theta} = H(z_1,z_2)^{-\theta} + 1.
\end{align*}
Thus it is easy to see that $G_1$ and $G_2$ are identified up to  a location constant\footnote{This location constant plays the role of $k$  in part 4 of Theorem \ref{pro:locident2hurdle}.}.
$\qed$
\end{example}

\medskip

Once $Q_1(\bm{Z})$ and $Q_2(\bm{Z})$ are identified,  then under our assumptions we identify the joint density by 
\begin{equation}
\label{eq:MHf}
f_{V_1,V_2}(q_1,q_2) = \frac{\partial^2 \Pr[D=1\vert Q_1(\bm{Z}) = q_1, Q_2(\bm{Z})=q_2]}{\partial q_1  \partial q_2},
\end{equation}
and the marginal treatment effect is given by
\begin{equation}
\label{eq:MHMTE}
E(Y_1-Y_0\vert V_1=q_1, V_2=q_2)f_{V_1,V_2}(q_1,q_2)=\frac{\partial^2 E[Y\vert Q_1(\bm{Z}) = q_1, Q_2(\bm{Z})=q_2]}{\partial q_1  \partial q_2}.
\end{equation}
 Furthermore, it follows from Corollary \ref{iden:ATE-ATT} that  under the additional Assumption~\ref{ass:global},  the ATE, ATT and PRTE parameters are identified as well.

\subsection{Dynamic Treatment} 
\label{sub:dynamic-treatment}
To conclude our examples, let us consider a two-period model of dynamic treatment where treatment assignment $D^2$ in the second period depends on the first-period treatment $D^1$ and  outcome $Y^1$:
\[
D^1=\uniset(V_1<Q_1(\bm{Z}^1))
 \; \mbox{ and } \;
  D^2=\uniset\left(V_2<Q_2(\bm{Z}^2, D^1, Y^1)\right).
  \]
 The analyst observes $(D^1,D^2,Y^1,Y^2,\bm{Z}^1,\bm{Z}^2)$. Theorem~\ref{thm:iden} applies to this model, provided only that the functions $Q_1$ and $Q_2$ are identified. The identification of $Q_1$ is straightforward. To identify $Q_2$, we use the results of \cite{shaikhvytlacil2011}, which considers a model similar to our second-period treatment assignment. While they stress partial identification, their   Remark 2.2 (p.\ 954) gives a sufficient condition for point identification. Translated in our notation, this requires that 
 \begin{enumerate}
 	\item the support of $(\bm{Z}^2, Q_1(\bm{Z}^1))$ is the product of the support of $\bm{Z}^2$ and the support of 
	$Q_1(\bm{Z}^1)$, and that
	
	\item for every value $(\bm{z}^2,y^1)$ of $(\bm{Z}^2,Y^1)$, there is a value $\bm{\bar{z}}^2$ such that
$Q_2(\bm{\bar{z}}^2, 1, y^1)=Q_2(\bm{z}^2, 0, y^1)$; and there is a value $\bm{\underline{z}}^2$ such that
$Q_2(\bm{\underline{z}}^2, 0, y^1)=Q_2(\bm{z}^2, 1, y^1).$

 \end{enumerate}
 Assumption 1 above requires that the set of instruments in the second period has a component that does not affect treatment in the first period, and whose range of variation does not depend on the propensity score of the first period. Assumption~2 adds the requirement that the ranges of the second-period propensity scores are independent of the first-period treatment, for all values of the first-period outcome. 
 
 These  assumptions  require overlap between treatment branches. They would not hold, for instance, in a medical trial when patients are oriented towards completely different  treatments depending on how they fare early on.

\section{Relation to the Existing Literature}\label{sec:literature}

The existing literature is very large; we only discuss here the most directly relevant papers.

\subsection{Ordered Treatments with Discrete Instruments} 

\cite{angrist1995two} consider two-stage least-squares (TSLS) estimation of a model in which  the ordered treatment  takes a finite number of values, and a discrete-valued instrument is available.   They show that the TSLS estimator obtained by regressing outcome $Y$
on a preestimated $E(D\vert Z)$ converges to a weighted sum of {\em average causal responses\/} under some monotonicity assumption.
\cite{HUV2006, HUV2008} go beyond \cite{angrist1995two} by showing how the TSLS estimate can be reinterpreted in more transparent ways in the MTE framework. They also  analyze a family of discrete choice models, to which we now turn.

\subsection{Discrete Choice Models} 
\label{sub:discrete_choice_models}

\citet[see also \citet{HV2007-handbook}]{HUV2008} consider a multinomial discrete choice model of treatment. They posit 
\begin{align}\label{DCM-general}
D=k \iff R_k(\bm{Z})-U_k > R_l(\bm{Z})-U_l \mbox{ for } l=0, \ldots,K-1  
 \text{ such that $l \neq k$},
\end{align}
where the $U$'s are continuously distributed and independent of $\bm{Z}.$
Then they study the identification of marginal and local average treatment effects under assumptions that are similar to ours: continuous instruments that generate enough dimensions of variation in the thresholds.

As they note, the discrete choice model with an additive structure  implicitly imposes  monotonicity, in the following form: if the instruments $\bm{Z}$ change in a way that increases $R_k\parZ$ relative to all other $R_l\parZ$, then no observation with treatment value $k$ is assigned to a different treatment. 
 We make no such assumption, as  Example~\ref{case:example1}   and Figure~\ref{fig:ex1}  illustrate.  Our results extend those of \cite{HUV2008} to any model with identified thresholds.
We consider a discrete choice model with three alternatives as an example.

\begin{example}[Discrete Choice Model with Three Alternatives]
Suppose that $\mathcal{K} = \{0,1,2\}$ with $K=3$. 
Let 
$\tilde{R}_{0,1}(\bm{Z}) = R_0(\bm{Z}) - R_1(\bm{Z})$,
$\tilde{R}_{0,2}(\bm{Z}) = R_0(\bm{Z}) - R_2(\bm{Z})$
and
$\tilde{R}_{1,2}(\bm{Z}) = R_1(\bm{Z}) - R_2(\bm{Z})$.
Similarly, let
$\tilde{U}_{0,1} = U_0 - U_1$,
$\tilde{U}_{0,2} = U_0 - U_2$
and
$\tilde{U}_{1,2} = U_1 - U_2$. \
Let $V_{0,1} = F_{\tilde{U}_{0,1}}( \tilde{U}_{0,1})$ and 
$Q_{0,1}(\bm{Z}) = F_{\tilde{U}_{0,1}}( \tilde{R}_{0,1}(\bm{Z}))$.
Define $V_{0,2}$, $V_{1,2}$, $Q_{0,2}(\bm{Z})$ and $Q_{1,2}(\bm{Z})$  similarly. 
Then the selection mechanism in \eqref{DCM-general} can be rewritten as
\begin{itemize}
\item $D=0$ iff $V_{0,1} < Q_{0,1}(\bm{Z})$ and $V_{0,2} < Q_{0,2}(\bm{Z})$
\item $D=1$ iff $V_{0,1} > Q_{0,1}(\bm{Z})$ and $V_{1,2} < Q_{1,2}(\bm{Z})$
\item $D=2$ iff $V_{0,2} > Q_{0,2}(\bm{Z})$ and $V_{1,2} > Q_{1,2}(\bm{Z})$.
\end{itemize}
Our general result in Section \ref{sec:general-results} applies immediately once 
the $Q_{j,k}$'s are identified.  This can be done, for example,  by applying the results of \cite{matzkin1993-JoE, matzkin2007-advances}. $\qed$
\end{example}

There is a growing empirical literature on multivalued unordered treatments.
\cite{dahl2002} develops a semiparametric Roy model for migration across U.S. states. 
In his empirical work, the number of unordered treatment is 51 (50 states plus the District of Columbia) and he controls for selection bias by conditioning on migration probabilities.
\cite{kirkeboen2016}  use discrete instruments to obtain TSLS estimates of returns to different fields of study in postsecondary education in Norway. In their setup, the unordered treatments are different fields of study.
\cite{kline2016} use data from the Head Start Impact Study to estimate a semiparametric selection model. Their model has  three treatment cells: Head Start, competing preschool programs, and no preschool (that is, home care).  

Broadly speaking, these  papers are in the same vein as  Roy models and discrete choice models.
Our approach complements this literature by focusing on the role of unobserved heterogeneity and the selection mechanism.


\subsection{Unordered Monotonicity} 
\label{sub:unordered_monotonicity}
In  an important recent paper, \cite{heckmanpinto-pdt} introduce a new concept of monotonicity.   Their ``unordered monotonicity'' assumption can be rephrased in our notation in the following way. Take two values  $\bm{z}$ and $\bm{z}^\prime$ of the instruments $\bm{Z}$   and any treatment value  $k$.

\begin{assumption}[Unordered Monotonicity]\label{ass:unordered-monotonicity}	
 Denote $d_k(\bm{v},\bm{z})$ and $d_k(\bm{v},\bm{z^\prime})$ the counterfactual values of the variable $d_k=\uniset(D=k)$ for an observation with unobserved heterogeneity $\bm{v}$. Then 
\begin{align*}
 & d_k(\bm{v},\bm{z}) \geq d_k(\bm{v},\bm{z^\prime}) \; \forall \bm{v}; \\
 \mbox{ or: } &d_k(\bm{v},\bm{z}) \leq d_k(\bm{v},\bm{z^\prime}) \;  \forall \bm{v} .
\end{align*}
\end{assumption}
Unordered monotonicity  for treatment value $k$ requires that if some observations move out of (resp.\ into) treatment value $k$ when instruments change value from $\bm{z}$ to $\bm{z}^\prime$, then no observation can move  into (resp.\ out of) treatment value $k$.  For binary treatments, unordered monotonicity is equivalent to the usual monotonicity assumption: there cannot be both compliers and defiers. 
 When $K>2$, it is weaker than ordered choice.
For example, suppose that there are three options $\{0, 1, 2\}$ and that a change of instruments makes option 1 less appealing. 
Under ordered choice, all agents who give up option 1 must fall back on  option~0, or all must fall back on option~2. 
Unordered monotonicity allows different agents to fall back on different options.  
 It still rules out  two-way flows, that is agents moving from option 0 or 2 into option 1.

\cite{heckmanpinto-pdt}  show that unordered monotonicity   (for well-chosen changes in instruments) is  essentially equivalent to a treatment model based on rules that are additively separable in the unobserved variables---that is, the model of section~\ref{sub:discrete_choice_models}.
   In this interpretation,  changes in instruments that increase the mean utility of an alternative relative to all others are unordered monotonic for that alternative, for instance. We refer the reader to Section~6 of \cite{heckmanpinto-pdt}  for a more rigorous discussion, and to \cite{pinto-jmp} for an application to the Moving to Opportunity program.

Unlike us, \cite{heckmanpinto-pdt} do not require continuous instruments; all of their analysis is framed in terms of  discrete-valued instruments and treatments. Beyond this (important) difference, unordered monotonicity clearly obeys our assumptions. On the other hand, we allow for much more general models of treatment.  It would be impossible, for instance, to rewrite our Examples~\ref{case:example1}, \ref{case:example2} and \ref{case:example3} so that they obey unordered monotonicity. We illustrate this point using Example \ref{case:example1} below.

\begin{example*1}[continued]
In Example \ref{case:example1}, 
$D=2$  iff $(V_1-Q_1(\bm{Z}))$ and $(V_2-Q_2(\bm{Z}))$ have opposite signs.
 Note that  there are two unobserved categories  within $D=2$:
\begin{align*}
D=2a \; &\text{ iff } \;  V_1<Q_1  \; \mbox{ and } \; V_2>Q_2, \\
D=2b \; &\text{ iff } \; V_1>Q_1  \; \mbox{ and } \; V_2<Q_2.
\end{align*}
Each one is unordered monotonic; but because we only observe their union, $D=2$ is not unordered monotonic---increasing $Q_1$ brings more people into $2a$ but moves some out of $2b$, so that in the end we have two-way flows, contradicting unordered monotonicity. To put it differently, the selection mechanism in 
Example~\ref{case:example1} becomes 
a discrete choice model when  each of four  alternatives $d=0,1,2a,2b$ is observed; however,
we only observe whether alternative $d=0$, $d=1$ or $d = 2$  is chosen
in 
Example~\ref{case:example1}.
 This amounts to an  unordered monotonic treatment that is observed through a coarser  information partition; this coarsening  destroys  unordered monotonicity.
  $\qed$
\end{example*1}

\subsection{Other Nonmonotonic Models} 
\label{sub:other_nonmonotonic_models}

It is also worth commenting on other papers that break monotonicity. 
\cite{gautier-hoderlein} consider a triangular random coefficients model for the binary treatment case.  Their model  is motivated by a single agent Roy model with random coefficients.  Its selection mechanism is governed by
\[
D = 1\{ V_1 - Z_1 - g(Z_1,\ldots,Z_L) - \sum_{j=2}^J V_j f_j (Z_j) > 0 \}, 
\]
where $\bm{V} = (V_1, \ldots, V_J)$ is a vector of unobserved random variables,  $\bm{Z} = (Z_1,\ldots,Z_J)$ is a vector of instruments that are independent of $(Y_0,Y_1,\bm{V})$,  and the functions $f_2,\ldots,f_J$ and $g$ are unknown. 
If we limit our attention to the case of two unobservables as in the double hurdle model, then the selection equation in \cite{gautier-hoderlein} reduces to 
\[
D = 1\{ V_1 - Z_1 - g(Z_1,Z_2) -  V_2 f_2 (Z_2) > 0 \}.
\]
 Here changes in $Z_1$ conform to  monotonicity; but changes in $Z_2$ need not.  

\cite{Lewbel2016}  consider a different non-monotonic selection mechanism for estimating the average treatment effect. They show that the average treatment effect is identified when a binary treatment is assigned by 
 \[
D=\uniset\left( \alpha_0\leq Z + V \leq \alpha_1 \right), 
\]
where  $V$ is an unobserved random variable; $Z$ is a continuous  variable that satisfies
$E(Y_j\vert V,Z) = E(Y_j\vert V)$ for $j=0,1$ and $V \Perp Z$; and 
$\alpha_0, \alpha_1$ are unknown parameters.


\subsection{Models with Continuous Treatment}
 
\cite{Chesher:03} develops conditions to 
 identify derivatives of structural functions in nonseparable models by functionals of quantile regression functions. 
In addition, 
\cite{FHMV2008} consider  a potential outcome model with a continuous  treatment.  They assume a stochastic polynomial restriction and show that the average treatment effect  can be identified if
a suitable control function can be constructed using instruments. 

\cite{imbens2009identification} also consider selection on unobservables with a continuous treatment.  They assume that the treatment (more generally in their paper, an endogenous variable)
is given by $D=g(Z,V)$, with  $g$ increasing in a scalar unobserved $V$.
They  identify the average structural function as well as  quantile, average, and policy effects.
Other more recent identification results along this line  can be found
in  \cite{Torgovitsky2015} and \cite{DF2015} among others. 
One key restriction in this group of papers is the monotonicity in the scalar $V$ in the selection equation. 
We do not rely on this type of restriction, but we only focus on the case of multivalued treatments.
Hence, our  approach and those of
the papers cited in this subsection are complementary.

Finally, our approach shares some features with \cite{HoderleinMammen2007}.  They consider
the identification of marginal effects in nonseparable models without monotonicity. They show how local average structural derivatives can be identified. Like ours, their
 approach relies on differentiation of observed functionals. The parameters of interest  they study are quite different, however,
 and their selection mechanism  is not as explicit as ours.



\section{Proof of Theorem~\ref{thm:iden}}\label{sec:proof-iden}

	Our proof has three steps. We first write conditional moments as integrals with respect to indicator functions. Then we show that these integrals  are differentiable and we compute their multidimensional derivatives. Finally, we impose Assumption~\ref{ass:index} and we derive the equalities in the theorem.

	\medskip

	{\bf Step 1:}

Under the assumptions imposed in the theorem, for any $\bm{q}$ in the range of $\bm{Q}$, 
\begin{align*}
\lefteqn{ E[ G(Y) D_k\vert\bm{Q}\parZ = \bm{q}] }  \\
&= E[G(Y_k)\vert D = k, \bm{Q}\parZ = \bm{q}]  \Pr (D = k \vert \bm{Q}\parZ = \bm{q}) \\ 
&= E[G(Y_k) \vert d_k(\bm{V},\bm{Q}\parZ)=1,\bm{Q}\parZ = \bm{q}]  \Pr (d_k(\bm{V},\bm{Q}\parZ)=1 \vert \bm{Q}\parZ = \bm{q}) \\ 
&= E[G(Y_k) \vert d_k(\bm{V},\bm{q})=1 ]  \Pr \left(d_k(\bm{V},\bm{q}) =1\right) \\
&= E[G(Y_k) \uniset\left(d_k(\bm{V},\bm{q})=1\right) ]\\
& =  E\left(E[G(Y_k) \uniset\left(d_k(\bm{V},\bm{q})=1\right)\vert \bm{V}]\right)
\\
& =  E\left(E[G(Y_k) \vert \bm{V}]\uniset\left(d_k(\bm{V},\bm{q})=1\right)\right),
\end{align*}
where the third equality follows from Assumption~\ref{ass:condindZ} and the others are obvious.
As a consequence,
\begin{align}
& \lefteqn{ E[ G(Y) D_k\vert\bm{Q}\parZ = \bm{q}] }\nonumber\\
&=  \int \uniset\left(d_k(\bm{v},\bm{q}) =1\right) E[G(Y_k) \vert \bm{V} =\bm{v}] f_{\bm{V}}(\bm{v}) 
d\bm{v}. \label{eq:qandp}
\end{align}

Let $b_k(\bm{v}) \equiv E[G(Y_k) \vert \bm{V} =\bm{v}] f_{\bm{V}}(\bm{v})$ and $B_k(\bm{q})=E[ G(Y) D_k\vert\bm{Q}\parZ = \bm{q}].$ Then
\eqref{eq:qandp} takes the form
\[
B_k(\bm{q}) = \int \uniset(d_k(\bm{v},\bm{q})=1) b_k(\bm{v})d\bm{v}.
\]
 Now recall from Lemma~\ref{lem:pikn}  that the indicator function of $D=k$ is a multivariate polynomial of the  indicator functions $S_j$ for $j\in\bm{J}$. Moreover,
 \[
S_j(\bm{V}, \bm{Q}\parZ)=\uniset(V_j<Q_j\parZ)=H(Q_j\parZ-V_j),
 \]
  where $H(t)=\uniset(t>0)$ is the one-dimensional Heaviside function. Therefore we can rewrite the selection of treatment $k$  as 
\begin{equation}
	\label{g-form-hf}
\uniset(d_k(\bm{v},\bm{q})=1) = 
\sum_{l \in \mathcal{L}} c^k_l  \prod_{j\in l} H(q_{j}-v_{j})
\end{equation}
and it follows that
\begin{equation}
	\label{eq:expandl}
	B_k(\bm{q})=\sum_{l \in \mathcal{L}} c^k_l  \int \left(\prod_{j\in l} H(q_{j}-v_{j})\right) b_k(\bm{v})d\bm{v}.
\end{equation}

\medskip

{\bf Step 2:}

By Assumption~\ref{ass:Q-continuity}, the  function $\bm{b}$ is locally equicontinuous;  and by Assumption~\ref{ass:open}, it is defined over an open neighborhood of $\bm{q}$.  This implies that all  terms in~\eqref{eq:expandl} are differentiable along all dimensions of  $\bm{q}.$ To see this, start with dimension~$j=1$. Any term $l$ in \eqref{eq:expandl}  that   does not contain $1$ is constant in $q_1$ and obviously differentiable. Take any other term and rewrite it as
\[
A_l(q_1)\equiv c^k_l \int_{0}^{q_1}  \int \left(\prod_{j\in l} H(q_{j}-v_{j})\right) b_k(v_1,\bm{v}_{-1})d\bm{v}_{-1} dv_1,
\]
where $\bm{v}_{-1}$ collects all directions of $\bm{v}$   in $l-\{1\}$.

Then for any $\varepsilon\neq 0$,
\begin{align*}
\frac{A_l(q_1+\varepsilon)-A_l(q_1)}{\varepsilon} &-c^k_l \int \left(\prod_{j\in l-\{1\}} H(q_{j}-v_{j})\right) b_k(q_1,\bm{v}_{-1})d\bm{v}_{-1} \\
&=
\frac{c^k_l}{\varepsilon} 
\int_{q_1}^{q_1 + \varepsilon}
\int \left(\prod_{j\in l-\{1\}} H(q_{j}-v_{j})\right) \left(b_k(v_1,\bm{v}_{-1})-b_k(q_1,\bm{v}_{-1})\right)d\bm{v}_{-1} dv_1.
\end{align*}

Since the functions $(b_k(\cdot,\bm{v}_{-1}))$ are locally equicontinuous at $q_1$,  for any $\eta>0$ we can choose $\varepsilon$ such that if $\abs{q_1-v_1}<\varepsilon,$
\[
\abs{b_k(q_1,\bm{v}_{-1})-b_k(v_1,\bm{v}_{-1})}<\eta;
\]
and since the Heaviside functions are bounded above by one, we  have
\[
\abs{\frac{A_l(q_1+\varepsilon)-A_l(q_1)}{\varepsilon} -c^k_l \int \left(\prod_{j\in l-\{1\}} H(q_{j}-v_{j})\right) b_k(q_1,\bm{v}_{-1})d\bm{v}_{-1}}<  \abs{c_l} \eta.
\]

This proves that  $A_l$ is differentiable in $q_1$ and that its derivative with respect to  $q_1$, which we denote   $A^1_l$, is 
\[
A^1_l=c_l \int \prod_{j \in l-\{1\}} H(q_j-v_j) \; b_k(q_1,\bm{v}_{-1})d\bm{v}_{-1}.
\]

But this derivative itself has the same form as $A_l$. Letting  $\bm{v}_{-{1,2}}$ collect all components of $\bm{v}$ except $(q_1,q_2)$, the same argument would prove that since the functions $(b_k(\cdot,\bm{v}_{-{1,2}}))$ are locally equicontinuous at $(q_1,q_2)$,  the function $A^1_l$ is differentiable with respect to $q_2$ and its derivative is
\[
c^k_l \int \left(\prod_{j\in l-\{1,2\}} H(q_{j}-v_{j})\right) \; b_k(q_1,q_2,\bm{v}_{-{1,2}})d\bm{v}_{-{1,2}}.
\]
Continuing this argument finally gives us the cross-derivative  with respect to $(\bm{q}^{l})$ as
\[
c^k_l \int b_k(\bm{q}^{l},\bm{v}_{-l})d\bm{v}_{-l},
\]
where 
$\bm{v}_{-l}$ collects all components of $\bm{v}$ whose indices are not in $l$.

\medskip

{\bf Step 3:}

Lemma~\ref{lem:pikn} and Assumption~\ref{ass:index} also imply that the    leading term  in the sum $\sum_l c^l_l \prod_{j\in l} H(q_{j}-v_{j})$ is
\[
c^k_{\bm{J}} \prod_{j=1}^J H(q_j-v_j).
\]

Now take the $J$-order derivative of $B(\bm{q})$ with respect to all  $q_j$ in turn.  By Lemma~\ref{lem:pikn}, the highest-degree term  of $B$ in $\bm{q}$   is
\[
c^k_{\bm{J}} \int \left(\prod_{j=1}^{J} H(q_j-v_j)\right) b_k(\bm{v})d\bm{v}
\]
as $c^k_{\bm{J}}\neq 0$ under Assumption~\ref{ass:index}; all other terms have a smaller number of indices $j$. 

This term contributes a cross-derivative
\[
c^k_{\bm{J}} b_k(\bm{q}),
\]
and all other terms generate   zero-value contributions  since each of them is constant in at least one of the directions $j$.

More formally, 
\begin{align}\label{important-derivatives}
T B_k(\bm{q})=\frac{\partial^J B_k(\bm{q})}{\prod_{j\in \bm{J}}\partial q_j}=c^k_{\bm{J}} b_k(\bm{q}).
\end{align}

\medskip

Given Assumptions \ref{ass:Q-continuity}  and \ref{ass:open},   we can apply 
  \eqref{important-derivatives}  successively to the pair of functions  
  \[
B_k(\bm{q})= E[G(Y)D_k\vert\bm{Q}\parZ = \bm{q}] 
\ \ \text{and} \ \
b_k(\bm{v})=E[G(Y_k) \vert \bm{V} =\bm{v}] f_{\bm{V}}(\bm{v}),
\]
 as in \eqref{eq:expandl}, and to the pair of functions
\[
B_k(\bm{q})= \Pr[D=k\vert\bm{Q}\parZ = \bm{q}] 
\ \ \text{with} \ \
b_k(\bm{v})=f_{\bm{V}}(\bm{v}).
\]
The first pair gives us the second equality in the Theorem, and the second pair gives us the first equality. 
$\qed$

{\singlespacing
\bibliographystyle{economet}
\bibliography{citations-multipleTreatments}
}

\newpage

\appendix

\renewcommand\thepage{A-\arabic{page}} \setcounter{page}{1}

\section*{Online Appendices to ``Identifying Effects of Multivalued Treatments''}

 Appendix~\ref{sec:appx-addresults} gives an identification result for the zero-index case, which was not dealt with in the text. It also provides a characterization of Heckman and Pinto's unordered monotonicity property as a subcase of our more general framework. Appendix~\ref{sec:appx-addproofs} collects proofs of  some of the results in the main text. Finally, Appendix~\ref{sec:the_entry_game} fills in the details of the entry game introduced in Section~\ref{sec:model}, and Appendix~\ref{appx:HUV2008}  compares  our results with  those of \cite{HUV2008} in more detail. Appendix \ref{sec:nonrect} discusses a more general form of threshold conditions than  the ``rectangular'' threshold conditions in Assumption~\ref{assumption:selection}.

\section{Additional Results}
\label{sec:appx-addresults}

\subsection{Identification with a Zero Index} 
\label{sub:index-zero-case}

Theorem~\ref{thm:iden} required that  the index of  treatment  $k$ be non-zero (Assumption \ref{ass:index}). It therefore  does not apply to, for instance, Example~\ref{case:example3}. Recall that in that example, 
\[
D_0 = \mathcal{D}_0(\bm{S}) = 1-S_1-S_2-S_3+S_1S_2+S_1S_3+S_2S_3
\]
and treatment 0 has degree $m^0=2<J^0=3$.

Note, however, that steps 1 and 2 of the proof of Theorem~\ref{thm:iden} apply to zero-index treatments as well; the relevant polynomial of Heaviside functions  has leading term
\[
H(q_1-v_1)H(q_2-v_2)+H(q_1-v_1)H(q_3-v_3)+H(q_2-v_2)H(q_3-v_3),
\]
and we can take the derivative in $(q_1,q_2)$ for instance to obtain an equation that replaces ~\eqref{important-derivatives}:
\begin{align*}
\ddpart{}{q_1}{q_2} B_0(\bm{q})=\int b_0(q_1,q_2,v_3) dv_3.
\end{align*}
Applying this to 
 $B_0(\bm{q})=\Pr[D=0 \vert\bm{Q}\parZ = \bm{q}]$ and
$b_0(\bm{v})=f_{\bm{V}}(\bm{v})$, and then to 
 $B_0(\bm{q})= E[ Y D_0 \vert\bm{Q}\parZ = \bm{q}]$ and 
$b_0(\bm{v})=E[G(Y_0) \vert \bm{V} =\bm{v}]f_{\bm{V}}(\bm{v})$, identifies
\[
\int f_{V_1,V_2,V_3}(q_1,q_2,v_3)dv_3=f_{V_1,V_2}(v_1,v_2)
\]
and
\begin{multline*}
\int E[G(Y_0)\vert V_1=q_1,V_2=q_2, V_3=v_3]f_{V_1,V_2,V_3}(q_1,q_2,v_3)dv_3 \\
= E[G(Y_0)\vert V_1=q_1,V_2=q_2] f_{V_1,V_2}(v_1,v_2).
\end{multline*}
Dividing through identifies  a  local counterfactual outcome:
\[
E[G(Y_0)\vert V_1=q_1,V_2=q_2].
\]
Under  Assumption~\ref{ass:global}, this also identifies $EG(Y_0)$.
Moreover, we can apply the same logic to the pairs $(q_1,q_3)$ and $(q_2,q_3)$ to get further information on the treatment effects.

\medskip

This argument applies more generally.  It allows us to state the following theorem:

\begin{theorem}[Identification with a zero index]\label{thm:idenzero}
Let Assumptions \ref{assumption:selection}, \ref{ass:condindZ} and   \ref{ass:continuous} hold. Fix a value  $\bm{q}$ in $\mathcal{\tilde{Q}}$, so that Assumptions~\ref{ass:Q-continuity} and \ref{ass:open}  also hold at $\bm{q}$.
Let $m$ be the degree of treatment $k$.  Take $l$ to be any subset of $\bm{J}$ that corresponds to a leading term in the expansion of the indicator function of $\{D=k\}$.  Denote  $\widetilde{T}$ the differential operator
\[
\widetilde{T}=\frac{\partial^{m}}{\prod_{i=1,\ldots,m} \partial_{l_i}}.
\] 
Then for $\bm{q} = (\bm{q}^{l}, \bm{q}^{\bm{J} - l})$, 	
\begin{align*}
 f_{\bm{V}^{l}}(\bm{q}^{l}) 
&= \frac{1}{c^k_l}
\widetilde{T}\, \Pr[D=k\vert\bm{Q}\parZ = \bm{q}] \\[5mm]
 E[G(Y_k) \vert \bm{V}^{l} =\bm{q}^{l}] 
&= 
\frac{\widetilde{T}\, E[ G(Y)D_k\vert\bm{Q}\parZ = \bm{q}]}
{\widetilde{T}\, \Pr[D=k\vert\bm{Q}(\bm{Z})= \bm{q}]}.
\end{align*}

\end{theorem}

\begin{proof}[Proof of Theorem \ref{thm:idenzero}] 
The proof of Theorem \ref{thm:idenzero} is  basically the same as that of Theorem \ref{thm:iden}. Steps 1 and 2 of the proof of Theorem~\ref{thm:iden} do not rely on any assumption about indices. They show that if we define
\[
W_l(\bm{q})=\int \prod_{j \in  l} H(q_j-v_j) b_k(\bm{v})d\bm{v}
\]
where the set $l\subset \bm{J}$,
then
its cross-derivative  with respect to $(\bm{p}^{l})$ is
\[
 \int b_k(\bm{q}^{l},\bm{v}_{-l})d\bm{v}_{-l},
\]
where 
$\bm{v}_{-l}$ collects all components of $\bm{v}$ whose indices are not in $l$.

Now let $m$ be the degree of treatment $k$. In the sum~\eqref{eq:expandl},  take any term $l$ such that $\abs{l}=m$.   Recall that  $\widetilde{T}$ denotes the differential operator
\[
\widetilde{T}=\frac{\partial^{m}}{\prod_{i=1,\ldots,m} \partial_{j_i}}.
\] 
By the formula above, applying $\widetilde{T}$ to term  $l$ gives
\[
c_l  \int b_k(\bm{q}^{l},\bm{v}_{-l})d\bm{v}_{-l}.
\]
Moreover, applying $\widetilde{T}$ to any other term $l^\prime$ obviously gives zero if term $l^\prime$ has degree less than $m.$  Now take any other  term $l^\prime$ of degree $m$. As $\widetilde{T}$ takes at least one derivative along a direction that is not in $l^\prime$, that term must also contribute zero.

This proves that
\[
\widetilde{T} B_k(\bm{q})=c^k_l  \int b_k(\bm{q}^{l},\bm{v}_{-l})d\bm{v}_{-l};
\]
note that it also implies that $\widetilde{T} B_k(\bm{q})$ only depends on $\bm{q}^{l}$.

Applying this first to $b_k(\bm{v})=f_{\bm{V}}(\bm{v})$ and  $B_k(\bm{q})=\Pr(D=k\vert \bm{Q}\parZ=\bm{q})$, then to
$b_k(\bm{v})=E[G(Y_k) \vert \bm{V} =\bm{v}] f_{\bm{V}}(\bm{v})$ and $B_k(\bm{q})=E[G(Y)D_k\vert\bm{Q}\parZ = \bm{q}]$ exactly as in the proof of Theorem~\ref{thm:iden}, we get
\begin{align*}
	\int f_{\bm{V}}(\bm{q}^{l},\bm{v}_{-l})
	d\bm{v}_{-l}
	&=
	\frac{1}{c^k_l}\widetilde{T}\Pr(D=k\vert \bm{Q}\parZ=\bm{q})\\
	\int E[G(Y_k) \vert \bm{V} =(\bm{q}^{l},\bm{v}_{-l})] f_{\bm{V}}(\bm{q}^{l},\bm{v}_{-l})d\bm{v}_{-l} &=
	\frac{1}{c^k_l}\widetilde{T}E(G(Y)D_k\vert \bm{Q}\parZ=\bm{q}).
\end{align*}
Since the left-hand sides are simply $f_{\bm{V}^{l}}(\bm{v}^{l})$ and  $E[G(Y_k)\vert \bm{V}^{l}=\bm{q}^{l}]f_{\bm{V}^{l}}(\bm{v}^{l}),$ the conclusion of the Theorem follows immediately. $\qed$
\end{proof}

Theorem~\ref{thm:idenzero} is a generalization of Theorem~\ref{thm:iden} (just take $m=J$). It calls for three remarks. First, we could weaken its hypotheses somewhat.  We could for instance replace $(0,1)^{J}$ with  $(0,1)^{m}$ in the statement of  Assumption~\ref{ass:global}.

Second, when $m<J$ the treatment effects are overidentified. This is obvious from the equalities in
Theorem~\ref{thm:idenzero}, in which the right-hand side depends on $\bm{q}$ but the left-hand side only depends on $\bm{q}^{I}$.

 Finally, considering several treatment values can identify even more, since 
$\bm{V}$ is assumed to be the same across $k$. Theorem~\ref{thm:iden} would then imply that if there is any treatment value $k$ with a nonzero index, then the joint density $f_{\bm{V}}$ is identified from that treatment value.

\subsection{Further Analysis of Unordered Monotonicity} 
\label{sub:unordered_monotonicity_more_results}

Our formalism allows us to derive a new characterization of the unordered monotonicity property defined by \cite{heckmanpinto-pdt}. Take any treatment value $k$. In our model, a change in instruments  $\bm{Z}$ acts on the treatment assigned to an observation with unobserved characteristics $\bm{V}$ through the indicator functions $S_j=\uniset(V_j<Q_j(\bm{Z}))$, which depend on the thresholds $\bm{Q}\parZ.$

Unordered monotonicity requires that there exist changes in thresholds $\Delta \bm{Q}$ such that  for $\bm{Q}^\prime=\bm{Q+\Delta Q}$,
\[
\Pr\left\{ d_k(\bm{V},\bm{Q})=0 \mbox{ and } d_k(\bm{V},\bm{Q}^\prime)=1 \right\}
\times
\Pr\left\{ d_k(\bm{V},\bm{Q})=1 \mbox{ and } d_k(\bm{V},\bm{Q}^\prime)=0\right \} =0,
\]
where the probabilities are computed over the joint distribution of $\bm{V}.$

In our framework, several thresholds are  typically relevant for each treatment value. This makes the analysis of unordered monotonicity complex in general. To understand why,   we start from  the expression~\eqref{eq:dkpol} of $D_k$ as a polynomial of $\bm{S}=(S_1,\ldots,S_J)$ for $S_j(\bm{V},\bm{Q})=\uniset(V_j < Q_j)$.
For any change in thresholds $\Delta \bm{Q}$ that induces changes in the indicators $\Delta \bm{S},$  
Taylor's theorem yields
\begin{equation}
		\label{eq:expa2WF}
    \Delta D_k =\sum_{m=1}^J 
    \sum_{\alpha_1 + \ldots + \alpha_J = m} \frac{1}{ \alpha_1! \alpha_2! \cdots \alpha_J!}  \frac{\partial^m \mathcal{D}_k (\bm{S})}{\partial S_1^{\alpha_1} \partial S_2^{\alpha_2} \ldots\partial S_J^{\alpha_J}} \prod_{l=1}^J \Delta S_{l}^{\alpha_l},
\end{equation}
where $\alpha_j$ is a nonnegative integer for $j=1,\ldots,J$.
Note that this is an exact expansion since $\mathcal{D}_k$ is a polynomial.	Moreover,  note that given a change in one threshold $\Delta Q_j,$ only $S_j$ changes and 
  \begin{equation}
  		\label{eq:deltaSj}
    \Delta S_j=\uniset(0 < V_j-Q_j < \Delta Q_j) - \uniset(\Delta Q_j < V_j-Q_j < 0).
   \end{equation}
(We do not need to distinguish between the weak and strict inequalities since 
	the distribution of $V_j$ is absolutely continuous with respect to the Lebesgue measure.)

The changes $\Delta S_j$ can only take the values 0 or $\pm 1$. In general higher-order terms in expansion~\ref{eq:expa2WF} may be nonzero. However, if the changes in thresholds $\Delta \bm{Q}$ are small then we can neglect the higher order terms since  the  values of $\bm{V}$ for which several $\Delta S_j$ are nonzero  occur with very small probability. To make this more precise, we use the following definition:

\begin{definition}[Two-Way Flows]\label{def:first-order-flows}
	A change in thresholds $\Delta\bm{Q}$ generates two-way flows  for treatment value $k$ if and only if
	\[
	\lim_{\varepsilon\to 0}
	\left(
	\frac{\Pr\left(D_k(0)=0 \mbox{ and } D_k(\varepsilon)=1\right)}{\varepsilon}
	\times
	\frac{\Pr\left(D_k(0)=1 \mbox{ and } D_k(\varepsilon)=0\right)}{\varepsilon}
	\right) > 0
	\]
	for $D_k(\varepsilon)\equiv d_k(\bm{V},\bm{Q}+\varepsilon\Delta\bm{Q}).$
\end{definition}

We now provide new characterizations of unordered monotonicity. 
 
\begin{theorem}[Characterizing Unordered Monotonicity in the Small]\label{thm:2WF}
        Fix a value $\bm{Q}$  of the thresholds. Denote
	\[
	\nabla\mathcal{D}_k(\bm{S})=\dpart{\mathcal{D}_k}{\bm{S}}(\bm{S}).
	\] 
	Assume that $J\geq 2$ and that there exist two values $j_1\neq j_2$ such that $\nabla_{j_1}\mathcal{D}_k$ and $\nabla_{j_2}\mathcal{D}_k$ are not identically zero. Then:
\begin{enumerate}
	\item If each component of $\nabla\mathcal{D}_k(\bm{S})$ has a constant sign when $\bm{S}$ varies over $\{0,1\}^{J}$, then some changes in thresholds  do not generate  two-way flows, and some others do.
	\item If the sign of any component $\nabla_j\mathcal{D}_k(\bm{S})$ changes  when $S_j$ switches between~0 and~1, then any change in thresholds  generates two-way flows.
\end{enumerate}	
(In these two statements,  we take 0 to have the same sign as  both $-1$ and  $+1$.)
\end{theorem}

\begin{proof}[Proof of Theorem \ref{thm:2WF}] 

Take $\varepsilon>0$ small. Remember that given a change in thresholds $\varepsilon \Delta Q_j,$
    \[
    \Delta S_j=\uniset(0 < V_j-Q_j < \varepsilon\Delta Q_j) - \uniset(\varepsilon\Delta Q_j < V_j-Q_j < 0),
	\]
	which is zero or has the sign of $\Delta Q_j$.
	
	Under our assumptions on the distribution of $\bm{V}$, the probability that $\Delta S_j\neq 0$ is of order $\varepsilon$;  the probability that 
	 $\Delta S_j\Delta S_l\neq 0$ is of order $\varepsilon^2$, etc. Given Definition~\ref{def:first-order-flows}, we only need to work on the first-order terms in expansion~\eqref{eq:expa2WF} since the other terms generate vanishingly small corrections. That is, we use
	\begin{align}
						\label{eq:deltaDk}
		\Delta D_k & \simeq \sum_{j=1}^J \nabla_j\mathcal{D}_k(\bm{S})\times \Delta S_j  
		\\
		& = \sum_{j=1}^J \nabla_j\mathcal{D}_k(\bm{S})\times \left(\uniset(0 < V_j-Q_j < \varepsilon\Delta Q_j) - \uniset(\varepsilon\Delta Q_j < V_j-Q_j < 0)\right). \nonumber
	\end{align}

	\begin{itemize}
		\item {\em Proof of part 1:}
		
		To prove part 1 of the theorem, assume that each derivative $\nabla_j\mathcal{D}_k$  has a constant sign, independent of $\bm{S}\in \{0,1\}^{J}$. 
		
		Then it is easy to find changes $\Delta\bm{Q}$ that only generate one-way flows. First take each $\Delta Q_j$ to have the sign of $\nabla_j\mathcal{D}_k$.
		
		 Since each $\Delta S_j$ has the sign of the corresponding $\Delta Q_j$, each product term in the sum~\eqref{eq:deltaDk} is non-negative, and so is the change in $D_k$. Obviously, changing the sign of all $\Delta Q_j$'s would generate one-way flows in the opposite direction.

		It is equally easy to find changes in instruments that generate two-way flows. Take the indices $j_1$ and $j_2$ referred to in the statement of the theorem.  Take $\Delta Q_m=0$ for $m\neq j_1,j_2$.  Then  expansion~\eqref{eq:deltaDk} becomes
		\[
		\Delta D_k \simeq\nabla_{j_1}\mathcal{D}_k(\bm{S})\times \Delta S_{j_1} +
		\nabla_{j_2}\mathcal{D}_k(\bm{S})\times\Delta S_{j_2}.
\]
		Choose some $\Delta Q_{j_1}, \Delta Q_{j_2}\neq 0$ such that 
		\[
		\nabla_{j_1}\mathcal{D}_k(\bm{S})\times\Delta Q_{j_1} \mbox{ and } \nabla_{j_2}\mathcal{D}_k(\bm{S})\times\Delta Q_{j_2}
		\]
		have opposite signs (which do not vary with $\bm{S}$ by assumption).
		
		Take $\abs{V_{j_1}-Q_{j_1}}$ small  and $\abs{V_{j_2}-Q_{j_2}}$ not small, so that $\Delta S_{j_1}$ has the sign of $\Delta Q_{j_1}$ and $\Delta S_{j_2}=0$; then  $\Delta D_k$ has the sign of  $\nabla_{j_1}\mathcal{D}_k(\bm{S})\times\Delta Q_{j_1}$. Permuting  ${j_1}$ and ${j_2}$ generates the opposite sign; therefore such a change in thresholds generates two-way flows.
			\item {\em Proof of part 2:}

			To prove part 2 of the theorem, take $j$ such that $\nabla_{j}\mathcal{D}_k$ changes sign when  the sign of $V_j-Q_j$ changes (so that $S_j$ switches between~0 and~1). Let $\Delta Q_m=0$ for all $m\neq j$, so that
			\[
			\Delta D_k\simeq \nabla_{j}\mathcal{D}_k(\bm{S})\times \Delta S_j.
			\]
			 By the assumption in part 2, the sign of $\Delta D_k$ is the sign of $\Delta S_j$ for some values of $\bm{V}$ and the opposite sign for other values.  Take any change in the threshold $\Delta Q_j$. Since $\Delta S_j$ is zero or has the sign of $\Delta Q_j$, 
			$\Delta D_k$ must take opposite values as $\bm{V}$ varies. $\qed$
	\end{itemize}

\end{proof}

To illustrate the theorem, first consider  
the double hurdle model, for which  $\nabla\mathcal{D}_1(\bm{S})=(S_2,S_1) \geq 0.$ 
This case is covered by part 1 of Theorem~\ref{thm:2WF}.
Changes such that $\Delta Q_1$ and $\Delta Q_2$ have the same sign do not generate two-way flows, but changes that generate $\Delta Q_1\Delta Q_2<0$ do. 

Now turn to the model of Example~\ref{case:example1}, where 
 $\nabla\mathcal{D}_2(\bm{S})=(1-2S_2, 1-2S_1).$ This corresponds to part~2 of the Theorem, since the sign of $(1-2s)$ depends on $s=0,1$. Using the expansion~\eqref{eq:deltaDk} gives, with $j_1=1, j_2=2$:
 \[
 \Delta D_2 \simeq (1-2S_2)\times\Delta S_{1} +
 (1-2S_1)\times\Delta S_{2}.
 \]
Depending on the values of $\bm{V}$ and therefore of $S_1$ and $S_2$, this can be
 \[
\Delta S_1+\Delta S_2, \Delta S_1-\Delta S_2, \Delta S_2-\Delta S_1, \mbox{ or }-\Delta S_1- \Delta S_2.
 \]
  To get one way flows only, we would need to chaneg thresholds to induce  $\Delta S_1, \Delta S_2=\pm 1$ such that the four numbers above have 
the same sign. But that is clearly impossible.
 Hence \emph{any}  change in instruments creates two-way flows.

 \section{Additional Proofs}
 \label{sec:appx-addproofs}
 
 \subsection{Proof of Corollary \ref{iden:ATE-ATT}}\label{appx:proof-iden-ATE-ATT}

First consider the average treatment effect. 
Under Assumption~\ref{ass:global}, we have that 
\[
E G(Y_k)=\int E\left(G(Y_k) \vert \bm{V} =\bm{v}\right) f_{\bm{V}}(\bm{v}) d\bm{v},
\]
which implies \eqref{iden-eq-ATE} immediately.

Now consider $E[G(Y_k) - G(Y_\ell)|D = k]$. 
Note that 
\begin{align*}
& E[G(Y_k) - G(Y_\ell)|D = k, \bm{Q}\parZ = \bm{q}] \\
&= 
E[G(Y_k) - G(Y_\ell)\vert d_k(\bm{V}, \bm{q}) =1 ] \\
&= \frac{\int \uniset\left(d_k(\bm{v},\bm{q}) =1\right) E[G(Y_k) - G(Y_\ell) \vert \bm{V} =\bm{v}] f_{\bm{V}}(\bm{v}) d\bm{v}}{\int \uniset\left(d_k(\bm{v},\bm{q}) =1\right) f_{\bm{V}}(\bm{v}) d\bm{v}}.
\end{align*}
Thus,
\begin{align*}
& E[G(Y_k) - G(Y_\ell)\vert D = k] \\
&= E E[G(Y_k) - G(Y_\ell)\vert D = k, \bm{Q}\parZ] ] \\
&= \int \frac{\int \uniset\left(d_k(\bm{v},\bm{q}) =1\right) E[G(Y_k) - G(Y_\ell) \vert \bm{V} =\bm{v}] f_{\bm{V}}(\bm{v}) d\bm{v}}{\int \uniset\left(d_k(\bm{v},\bm{q}) =1\right) f_{\bm{V}}(\bm{v}) d\bm{v}}
d F_{\bm{Q}\parZ| D} (\bm{q}|k).
\end{align*}
By Bayes' rule, we have that
\[
d F_{\bm{Q}\parZ| D} (\bm{q}|k) = \frac{\Pr[D=k|\bm{Q}\parZ = \bm{q}]}{\Pr(D=k)} d F_{\bm{Q}\parZ} (\bm{q}).
\]
Since 
\[
\Pr[D=k|\bm{Q}\parZ = \bm{q}] = \int \uniset\left(d_k(\bm{v},\bm{q}) =1\right) f_{\bm{V}}(\bm{v}) d\bm{v},
\]
we have that
\begin{align*}
& E[G(Y_k) - G(Y_\ell)|D = k] \\
&= \int \frac{\int \uniset\left(d_k(\bm{v},\bm{q}) =1\right) E[G(Y_k) - G(Y_\ell) \vert \bm{V} =\bm{v}] f_{\bm{V}}(\bm{v}) d\bm{v}}{\Pr(D=k)}
d F_{\bm{Q}\parZ} (\bm{q}) \\
&=  \frac{\int \Pr \left(d_k(\bm{v},\bm{Q}\parZ) =1\right) E[G(Y_k) - G(Y_\ell) \vert \bm{V} =\bm{v}] f_{\bm{V}}(\bm{v}) d\bm{v}}{\Pr(D=k)} \\
&= \int  \Delta_{\text{MTE}}^{(k, \ell)} (\bm{v})  \omega^k_{\text{TT}}(\bm{v})  d\bm{v}.
\end{align*}

We now move to the identification of the policy relevant treatment effects. 
Recall that in the proof of Theorem 3.1 (see equation \eqref{eq:qandp}), we have that 
\begin{align*}
& \lefteqn{ E[ G(Y) D_k\vert\bm{Q}\parZ = \bm{q}] }\nonumber\\
&=  \int \uniset\left(d_k(\bm{v},\bm{q}) =1\right) E[G(Y_k) \vert \bm{V} =\bm{v}] f_{\bm{V}}(\bm{v}) 
d\bm{v}.
\end{align*}
Since $G(Y)  = \sum_{k \in \mathcal{K}} G(Y) D_k$,  
we then have that
\begin{align*}
E[G(Y)] 
&= \sum_{k \in \mathcal{K}} E[ E[ G(Y) D_k\vert\bm{Q}\parZ = \bm{q}] ] \\
&=  \sum_{k \in \mathcal{K}} \int \Pr[ d_k(\bm{v},\bm{Q}\parZ) =1 ] E[G(Y_k) \vert \bm{V} =\bm{v}] f_{\bm{V}}(\bm{v}) 
d\bm{v}.
\end{align*}
Similarly, we have that
\begin{align*}
E[D] 
&= \sum_{k \in \mathcal{K}} k E[ E[ D_k\vert\bm{Q}\parZ = \bm{q}] ] \\
&=  \sum_{k \in \mathcal{K}} k \int \Pr[ d_k(\bm{v},\bm{Q}\parZ) =1 ] f_{\bm{V}}(\bm{v}) 
d\bm{v}
\end{align*}
and that 
\begin{align*}
E[D_k = 1] 
&=  E[ E[ D_k\vert\bm{Q}\parZ = \bm{q}] ] \\
&=   \int \Pr[ d_k(\bm{v},\bm{Q}\parZ) =1 ] f_{\bm{V}}(\bm{v}) 
d\bm{v}.
\end{align*}
Then the desired results follow immediately since the new policy only changes  $\bm{Q}$ to $\bm{Q^\ast}$, while everything else remains the same.
$\qed$

\subsection{Proof of Theorem \ref{pro:Q-two-way-flows}}\label{appendix:Q-two-way-flows}

It follows from \eqref{two-way-flow-propensity-score} on  page~\pageref{page:defEx1} that 
\begin{align}\label{two-way-flow-propensity-score-more}
Q_1\parZ +Q_2\parZ &= 2 P_0\parZ +  P_2\parZ.
\end{align}
The right hand side of \eqref{two-way-flow-propensity-score-more}
is identified directly from the data. Suppose that  
 $\tilde{Q}_1\parZ$ and $\tilde{Q}_2\parZ$ also satisfy  $\tilde{Q}_1\parZ + \tilde{Q}_2\parZ = 2 P_0\parZ +  P_2\parZ$, as well as Assumption~\ref{ass:Q:iden-two-way-flows}.
Then writing $\Delta_j\parZ = Q_j\parZ - \tilde{Q}_j\parZ$ $(j=1,2)$ gives
$\Delta_1\parZ =- \Delta_2\parZ$. 
But by Assumption~\ref{ass:Q:iden-two-way-flows}, $\Delta_1$ does not depend on $Z_2$, and $\Delta_2$ does not depend on $Z_1$. Therefore 
we must have $\tilde{Q}_1(Z_1) = Q_1(Z_1) + C$ and $\tilde{Q}_2(Z_2) = Q_2(Z_2) - C$,
where $C$ is a constant. 
This proves that  $Q_1$ and $Q_2$ are identified up to an additive constant.

Further,  take any $(z_1^0,z_2^0)\in \mathcal{Z}$.
If we take $Q_2(z_2) = P(z_1^0,z_2)-C_1^0$ for some constant $C_1^0$, then  by \eqref{two-way-flow-propensity-score-more},
\begin{align}\label{iden-q1}
Q_1(z_1) = P(z_1,z_2)-P(z_1^0,z_2)+C_1^0.
\end{align}
Since the right-hand side of \eqref{iden-q1} should not depend on $z_2$, we set 
	\begin{align*}
		Q_1(z_1) &= P(z_1,z_2^0)-P(z_1^0,z_2^0)+C_1^0\\
		Q_2(z_2) &= P(z_1^0,z_2)-C_1^0.
	\end{align*}
To describe the possible range of $C_1^0$, 
note that we require that 
\begin{align*}
\Pr (D = 0) &= \Pr[ Q_1(Z_1) > 0 \text{ and } Q_2(Z_2) > 0 ] > 0, \\
\Pr (D = 1) &= \Pr[ Q_1(Z_1) < 1 \text{ and } Q_2(Z_2) < 1 ] > 0, \\
\Pr (D = 2) &= \Pr[ Q_1(Z_1) > 0 \text{ and } Q_2(Z_2) < 1 ] +
\Pr[ Q_1(Z_1) <1  \text{ and } Q_2(Z_2) > 0 ] > 0.
\end{align*}
That is, $C_1^0$ must satisfy the following restrictions:
\begin{align*}
\begin{split}
& \Pr[ P(z_1^0,z_2^0) - P(Z_1,z_2^0) < C_1^0 < P(z_1^0,Z_2) ] > 0, \\
& \Pr[ P(z_1^0,Z_2) - 1 < C_1^0 < 1 + P(z_1^0,z_2^0) - P(Z_1,z_2^0) ] > 0, \\
& \Pr \Big[ \max \{ P(z_1^0,z_2^0) - P(Z_1,z_2^0), P(z_1^0,Z_2) - 1 \} < C_1^0 \Big ] \\
&+
\Pr \Big[  C_1^0 < \min \{ 1 + P(z_1^0,z_2^0) - P(Z_1,z_2^0), P(z_1^0,Z_2) \}  \Big] > 0.
\end{split}
\end{align*}

$\qed$

\subsection{Proof of Theorem~\ref{pro:globident2hurdle}} 
\label{sub:proof_of_theorem_ref_pro_globident2hurdle}

Recall that we denote $H(z_1,z_2)=\Pr(D=1\vert Z_1=z_1,Z_2=z_2)$ the  propensity score. Under our exclusion restrictions, $H(z_1,z_2)=F_{V_1,V_2}(G_1(z_1),G_2(z_2))$.

Let $f_{\bm{V}} (v_1, v_2)$ denote the density of $\bm{V}=(V_1,V_2)$.	By construction, 
\begin{align}\label{me1}	
H(z_1,z_2)
=F_{\bm{V}}(G_1(z_1), G_2(z_2))
= \int_{0}^{G_1(z_1)} \int_{0}^{G_2(z_2)} f_{\bm{V}} (v_1, v_2) dv_1 dv_2.
\end{align}
Differentiating both sides of \eqref{me1} with respect to $z_1$ gives	
\begin{align}\label{me2}
\frac{\partial H}{\partial z_1}(z_1,z_2)
= G^\prime_1(z_1)\int_{0}^{G_2(z_2)} f_{\bm{V}} (G_1(z_1), v_2) dv_2.
\end{align}	
Now letting $z_2 \rightarrow b_2$ on  both sides of  \eqref{me2} yields
\begin{align}\label{me3}
\lim_{z_2 \rightarrow b_2} \frac{\partial H}{\partial z_1}(z_1,z_2)
= G^\prime_1(z_1)\left[ \lim_{z_2 \rightarrow b_2} \int_{0}^{G_2(z_2)} f_{\bm{V}} (G_1(z_1), v_2) dv_2 \right]. 
\end{align}		
The expression inside the brackets on the right side of 
\eqref{me3} is 1 
since $\lim_{z_2 \rightarrow b_2} G_2(z_2) = 1$
and the marginal distribution of $V_2$ is $U[0,1]$. 	
Therefore we identify $G_1$ by
\begin{align}\label{me4}
G_1(z_1) = \int_{a_1}^{z_1} \lim_{t_2 \rightarrow b_2} \frac{\partial H}{\partial z_1}(t_1,t_2) dt_1.
\end{align}	
Analogously, we 	
identify $G_2$ by
\begin{align}\label{me4-3}
G_2(z_2) = \int_{a_2}^{z_2} \lim_{t_1 \rightarrow b_1} \frac{\partial H}{\partial z_2} (t_1,t_2)dt_2.
\end{align}	
Returning to \eqref{me1},	since $G_1$ and $G_2$ are strictly increasing 
we  identify $F_{\bm{V}}$ by 
\begin{align*}
F_{\bm{V}}(v_1, v_2) = 
H(G_1^{-1}(v_1), G_2^{-1}(v_2)).
\end{align*}
 $\qed$

\subsection{Proof of Theorem~\ref{pro:locident2hurdle}} 
\label{sub:proof_of_theorem_ref_pro_locident2hurdle}

\subsubsection{Proof of part 1} 
\label{ssub:proof_of_1}

Given our differentiability assumptions,  we can  take derivatives of the formula
\begin{equation}
	\label{eq:phiphi}
	 \phi\left(H(z_1,z_2)\right)=\phi(G_1(z_1))+\phi(G_2(z_2))
\end{equation}
over $\mathcal{N}$. 
Using 
 \[
 \ddpart{\left(\phi \circ H\right)}{z_1}{z_2}(z_1,z_2)=0,
 \]
we obtain
 \[
\phi^{\prime\prime}(h) \dpart{H}{z_1}(z_1,z_2)\dpart{H}{z_2}(z_1,z_2) +\phi^{\prime}(h) \ddpart{H}{z_1}{z_2}(z_1,z_2)=0
 \]
 with $h=H(z_1,z_2)$.
 
 Take any smooth curve  contained in $\mathcal{N}$   and parameterize it as $h \to (z_1(h),z_2(h))$ with $h=H(z_1(h), z_2(h))$; then we have a differential equation
\begin{equation}
	\label{eq:dphi}
 \phi^{\prime\prime}(h) \dpart{H}{z_1}(z_1(h),z_2(h))\dpart{H}{z_2}(z_1(h),z_2(h)) 
 +\phi^{\prime}(h) \ddpart{H}{z_1}{z_2}(z_1(h),z_2(h)) = 0.
\end{equation}
  Using  \eqref{eq:phiphi}, the partial derivatives $H_1$ and $H_2$ cannot take the value zero on $\mathcal{N}$ since $G^\prime_1$ and $G^\prime_2$ are never zero. Therefore we can  rewrite~\eqref{eq:dphi} as
 \[
 \frac{\phi^{\prime\prime}}{\phi^{\prime}}(h) = - \frac{H_{12}}{H_1 H_2}(z_1(h), z_2(h))
 \]
over $\mathcal{N}$. 

 We note that this equation incorporates a sign constraint and  overidentifying restrictions. For $\phi$ to be strictly decreasing and convex, we require $H_{12}/(H_1H_2)\geq 0$. Moreover, on any admissible curve the ratio 
$H_{12}/(H_1H_2)$
must be  the same function of $h$, which we denote $R(h)$. 
$\qed$


\medskip

\subsubsection{Proof of part 2} 
\label{ssub:proof_of_2}

From now on  we denote $(\underline{h},\overline{h})\subset (0,1)$ the image of   $\mathcal{N}$  by $H$.

We   use the fact that 
  $\partial \log(-\phi'(h))/\partial h = {\phi^{\prime\prime}(h)}/{\phi^{\prime}(h)}$ to obtain
 \[
 \log \left(-\phi^\prime(h)\right)=\int_{h}^{\bar{h}} R(t) dt+ \log \left(-\phi^\prime(\bar{h})\right),
 \] 
 so that
 \[
 \phi^\prime(h)=\phi^\prime(\bar{h})\exp\left(\int_{h}^{\bar{h}} R(t) dt\right).
 \]
Denoting 
 \[
 \mathbb{T}(h):=\int_{h}^{\bar{h}} dk \exp\left(\int_{k}^{\bar{h}} R(t) dt\right)
 \]
gives us $\phi(h)=\phi(\bar{h})-\phi^\prime(\bar{h})\mathbb{T}(h).$ Note that by construction $\mathbb{T}$ is a decreasing function and $\mathbb{T}(\bar{h})=0.$  Moreover, $\phi^\prime(\bar{h})$ cannot be zero since $\phi$ would be constant.
$\qed$

 
 \subsubsection{Proof of part 3} 
 \label{ssub:proof_of_3}

If $\phi$ solves \eqref{eq:phiphi} then clearly so does $\alpha\phi$ for any $\alpha>0$;  we normalize   $\phi^\prime(\bar{h})=-1$. 
Hence, from now on, $\phi(h)=\phi(\bar{h}) - \mathbb{T}(h).$
 The constant $\phi(\bar{h})$  must be non-negative since $\phi$ cannot take negative values. Moreover, since $\phi$ is convex, $\phi^\prime(\bar{h})=-1$, and $\phi(1)=0$, we must have $\phi(\bar{h})\leq 1-\bar{h}.$
 If moreover $\bar{h}=\sup_{\bm{z}\in\mathcal{N}}\Pr(D=1\vert \bm{Z}=\bm{z})=1$, then $\phi(\bar{h})=\phi(1)=0$; this defines directly $\phi(h)=-\mathbb{T}(h)$ over $(\underline{h},1)$.
$\qed$
 
 
 \subsubsection{Proof of part 4} 
 \label{ssub:proof_of_4}

Since the model is well-specified, there is a solution $G_1, G_2$ (the thresholds of the true DGP).
In addition, 
 since any other admissible $(\tilde{G}_1,\tilde{G}_2)$ must satisfy
	\[
	\phi(\tilde{G}_1(z_1))+\phi(\tilde{G}_2(z_2))=\phi(H(z_1,z_2))=\phi(G_1(z_1))+\phi(G_2(z_2))
	\]
	on $\mathcal{N}$, it must be that
	\begin{align*}
		\phi(\tilde{G}_1(z_1)) &= \phi(G_1(z_1))-k \\
		\phi(\tilde{G}_2(z_2)) &= \phi(G_2(z_2))+k
	\end{align*}
for some constant $k$. Any such constant must be such that 
$\phi(G_1(z_1))-k$ and $\phi(G_2(z_2))+k$ are both nonnegative for all $z_1$ and $z_2$ in the projections of $\mathcal{N}$. That is,
	\[
	-\inf \phi(G_2(z_2)) \leq k \leq \inf \phi(G_1(z_1)).
	\]
  If moreover $\sup_{\bm{z}\in \mathcal{N}} \Pr(D=1\vert \bm{Z}=\bm{z})=1$, then $\bar{h} = 1$.  Take a sequence $(\bm{z}_{n})$ such that $H(\bm{z}_n)$ converges to $\bar{h}=1$. Then $\phi(H(\bm{z}_n))$ converges to zero, so that both $\phi(G_1(z_{1n}))$ and $\phi(G_2(z_{2n}))$  must converge to zero. The double inequality above  implies that $k=0$, and 
  $G_1$ and $G_2$ are point-identified on the projections of $\mathcal{N}$. 
$\qed$


\section{The Entry Game} 
\label{sec:the_entry_game}

Let us return to Example~\ref{case:example2}, in which 
two firms $j=1,2$ are considering entry into a new market. Firm $j$ has profit $\pi_j^m$ if it becomes a monopoly, and $\pi^d_j<\pi^m_j$ if both firms enter. We saw that if  $\pi_j^m>0>\pi_j^d$ for both firms, then there are two symmetric equilibria, with only one firm operating. Now assume that we observe not only the number of entrants as in Example~\ref{case:example2}, but also their identity.  With profits 
given by  $\pi_j^m=V_j-Q_j(\bm{Z})$ and $\pi_j^d=\bar{V}_j-\bar{Q}_j(\bm{Z})$,  if only firm~1 entered then we know that $\pi_1^m>0$ and $\pi_2^d<0$, so that
\[
V_1>Q_1(\bm{Z})  \; \mbox{ and } \; \bar{V}_2<\bar{Q}_2(\bm{Z}).
\]
That still leaves two possible cases:
\begin{enumerate}
	\item $\pi^m_2<0$,  and the unique equilibrium has only firm 1 entering the market;
	\item and $\pi^m_2>0$, and there is another, symmetric equilibrium with only firm 2 entering.
\end{enumerate}
Now let us postulate an equilibrium selection rule that has a threshold structure: when both $\pi^m_1$ and $\pi^2_m$ are positive, firm~1 is selected to be the unique entrant if and only if $U<q\parZ$. Then the necessary and sufficient set of conditions for the entry of firm~1 only is
\[
	  V_1>Q_1(\bm{Z}) \mbox{ and } \left(V_2<Q_2(\bm{Z}) \mbox{ or } \left(\bar{V}_2<\bar{Q}_2(\bm{Z})
	  \mbox{ and } U<q\parZ\right)\right).
\]
This is again a special case of the general framework we analyze in this paper.

\section{Detailed Discussion of \cite{HUV2008}}
\label{appx:HUV2008}

 \cite{HUV2008}  consider a multinomial discrete choice model for treatment. 
 They posit 
\[
D=k \iff R_k(\bm{Z})-U_k > R_l(\bm{Z})-U_l \mbox{ for } l=0, \ldots,K-1  
 \text{ such that $l \neq k$},
\]
where the $U$'s are continuously distributed and independent of $\bm{Z}.$

Define 
\[
\bm{R}(\bm{Z})=\left(R_k\parZ-R_l\parZ\right)_{l\neq k}  \; \mbox{ and } \; 
\bm{U}=\left(U_k-U_l\right)_{l\neq k}. 
\]
Then $D_k=\uniset(\bm{R}(\bm{Z})> \bm{U})$; and defining
 $Q_l(\bm{Z})=\Pr[ \bm{U}_l <\bm{R}_l(\bm{Z})\vert\bm{Z}]$
 allows us to write  the treatment model as
\begin{align}
D=k  \; \mbox{ iff } \; \bm{V} <\bm{Q}\parZ,
\end{align}
where each  $V_l$ is distributed as $U[0,1].$

The applications they consider are GED certification (with three treatments: permanent high school dropout, GED, high school degree) and randomized trials with imperfect compliance (for example, no training, classroom training, and job search assistance).

They then study the identification of marginal and local average treatment effects under assumptions that are similar to ours: continuous instruments that generate enough dimensions of variation in the thresholds. 
 They assume that $\bm{V}$ is continuously distributed with full support; that $(\bm{U},\bm{V})\Perp Z$; and that all treatments have positive probabilities. More importantly, they make either
\begin{itemize}
	\item assumption (a): for each treatment $j$, there is a component of $\bm{Z}$ that drives some variation in $R_j$ conditional on the other components, and in $R_j$ only;
	\item assumption (b): for each treatment $j$, there is a component of $\bm{Z}$ that drives continuous variation in $R_j$ conditional on the other components, and no variation in the other components of $R$.
\end{itemize}

For any subset of treatments $\mathcal{J} \subset \mathcal{K}$, they define  $Y_{\mathcal{J}}$ to be the outcome when the agent chooses the best treatment from $\mathcal{J}$. 
They also define $\Delta_{\mathcal{J},\mathcal{L}}=Y_{\mathcal{J}}-Y_{\mathcal{L}}$, and in particular the MTE
\[
E \left(\Delta_{\mathcal{J},\mathcal{L}} \vert \bm{Z}, R_{\mathcal{J}}(\bm{Z})=R_{\mathcal{L}}(\bm{Z})\right).
\]
They show that
\begin{itemize}
	\item if we take $\mathcal{J}=\left\{j\right\}$ and $\mathcal{L}=\mathcal{K}-\left\{j\right\}$, then the LATE is identified under (a) and the MTE is identified under (b);
	\item if we take any $\mathcal{J}$ and $\mathcal{L}=\mathcal{K}-\mathcal{J}$, then the results are similar but the MTEs and LATEs are defined by conditioning on the values of the  $Q$'s rather than on the $Z$'s.
\end{itemize}
They do not invoke any large support assumptions to obtain identification results mentioned just above.

However, if we take $\mathcal{J}=\left\{j\right\}$ and $\mathcal{L}=\left\{l\right\}$, then their corresponding identification results (see Theorem 3 of  \cite{HUV2008}) require a large support condition.
To see their logic, suppose that $K=3$ and that one of the $R_j$'s is sufficiently negative  that the probability of choosing one of the choices is arbitrarily small. This case
 effectively reduces to the binary treatment case;  their LIV estimand, which is the limit of a sequence of Wald estimands,  identifies the MTE. 

We do not rely on this type of   identification-at-infinity strategy
since we identify the MTE via multidimensional cross derivatives.  
Note that our identification results are conditional on the assumption that $\bm{Q}$ is already identified. 
A more stringent assumption on the support of $\bm{Z}$ might be necessary to identify $\bm{Q}$, as demonstrated in \cite{matzkin1993-JoE, matzkin2007-advances}.
In this sense, our assumptions are not necessarily weaker than those of   \cite{HUV2008}.
We view  our identification results  and theirs as complementing each other.

\section{Non-rectangular Threshold Conditions}
\label{sec:nonrect}
The threshold conditions we postulated in Assumption~\ref{assumption:selection} have the ``rectangular'' form $V_j<Q_j(\bm{Z})$. 
Suppose that the threshold conditions $j=1,\ldots,J$ have the  more general form
\[
\bm{\alpha}_j\cdot \bm{U} \leq R_j(\bm{Z})
\]
where the $\bm{\alpha}_j$ are possibly unknown parameter  vectors in $\Real^L$ and $\bm{U}=(U_1,\ldots,U_L)$ is independent of $\bm{Z}$.  For notational simplicity, assume that each (scalar) random variable $u_j\equiv \bm{\alpha}_j\cdot \bm{U}$ has positive density everywhere; denote  $H_j$ its  cdf. Then each threshold condition can be written equivalently as
\[
V_j \equiv H_j(u_j) <  H_j(R_j(\bm{Z})) \equiv Q_j(\bm{Z}).
\]
By construction, each $V_j$ is distributed uniformly over $[0,1]$. Moreover, since each threshold $Q_j$ is an increasing function of the corresponding $R_j$ only, any exclusion restriction assumed on either form applies equally to the other, so that we can hope to identify the thresholds $Q_j$ under suitable assumptions. If they are indeed identified, then we can apply Theorem~\ref{thm:iden} to recover the joint density of $\bm{V}=(V_1,\ldots,V_j)$ and the MTE conditional on $\bm{v}$. 

The random variables $\bm{V}$ and the thresholds $\bm{Q}$ are only auxiliary objects, and the analyst is likely to be more interested in the $\bm{U}$  and $\bm{R}$. If the  cdf $H_j$ were known, then we could write $R_j=H^{-1}_j(Q_j)$ and by the change-of-variables formula,
\[
f_{\bm{u}}(u_1,\ldots,u_J) = f_{\bm{V}}\left(H^{-1}_1(u_1), \ldots, H^{-1}_J(u_J)\right) \times \prod_{j=1}^J H^\prime_j(u_j).
\]
In turn, knowing the joint distribution of $\bm{u}$ directly gives the density of $\bm{U}$ if $L=J$ and the  matrix $\bm{\alpha}$ whose rows are the vectors $\bm{\alpha}_j^\prime$ is invertible:
\[
f_{\bm{U}}(\bm{U}) = f_{\bm{u}}\left(\bm{\alpha} \bm{U}\right) \times \abs{\bm{\alpha}}.
\]

If more realistically the $H_j$ and $\bm{\alpha}_j$ are unknown, we may still use other restrictions.  As an illustration, take  a recursive system, where the matrix $\bm{\alpha}$ is lower-triangular with diagonal terms equal to one.  Then  since $U_2=u_2-\alpha_{21}u_1=H^{-1}_2(V_2)-\alpha_{21}H_1^{-1}(V_1)$, the independence of $U_1$ and $U_2$, for instance, would translate into the independence of $V_1$ and of the variable 
\[
W_2 \equiv H^{-1}_2(V_2)-\alpha_{21}H_1^{-1}(V_1).
\]
Now $V_2=H_2(W_2+\alpha_{21}U_1)$, so this in turn implies that the (identified)  distribution of $V_2$ conditional of $V_1$ must satisfy 
\[
F_{V_2\vert V_1}\left(H_2\left(w_2+\alpha_{21}H_1^{-1}(v_1)\right)\vert v_1\right) = F_{W_2}(w_2)=H_2(w_2)
\]
for all $w_2$ and $v_1$. But as the right-hand-side does not depend on $v_1$, this imposes restrictions that only hold for some choices of $H_1$, $H_2$ and $\alpha_{21}$.  If we only know $H_2$, then 
\[
w_2+\alpha_{21}H_1^{-1}(v_1) = F_{V_2\vert V_1}^{-1}\left(H_2(w_2)\vert v_1\right)
\]
overidentifies the product $\alpha_{21}H_1^{-1}(v_1)$; and if we also know  $H_1$, then it overidentifies $\alpha_{21}$. These results extend directly to higher-dimensional systems.

\end{document}